\documentclass[submission,copyright,creativecommons]{eptcs}
\providecommand{\event}{QPL 2023}

\input{_.sty}

\newcommand{\TITLE}{Graphical CSS Code Transformation Using ZX Calculus}
\title{\TITLE}

\author{Jiaxin Huang\footnote{equal authorship}
\institute{Dept. of Computer Science\\Hong Kong University of \\ Science and Technology}
\email{jhuangbo@connect.ust.hk}
\and
Sarah Meng Li\footnotemark[1]
\institute{Institute for Quantum Computing (IQC)\\Dept. of Combinatorics \& Optimization (C\&O) \\ University of Waterloo}
\email{\quad sarah.li@uwaterloo.ca}
\and
Lia Yeh\footnotemark[1]
\institute{University of Oxford}
\institute{Quantinuum\\ 17 Beaumont Street\\Oxford OX1 2NA, UK}
\email{\quad lia.yeh@cs.ox.ac.uk}
\and
Aleks Kissinger
\institute{University of Oxford}
\email{aleks.kissinger@cs.ox.ac.uk}\!\!\!\!\!\!\!
\and
Michele Mosca
\institute{Perimeter Institute (PI)\\ IQC, C\&O, University of Waterloo}
\email{michele.mosca@uwaterloo.ca}\!\!\!\!\!
\and
Michael Vasmer
\institute{PI, IQC, University of Waterloo}
\email{mvasmer@perimeterinstitute.ca}}

\def\titlerunning{\TITLE}
\def\authorrunning{J. Huang, S. M. Li, L. Yeh, A. Kissinger, M. Mosca, M. Vasmer}
\def\copyrightholders{J. Huang, S. M. Li, L. Yeh,\\A. Kissinger, M. Mosca \& M. Vasmer}

\begin{document}
\maketitle

\begin{abstract}
In this work, we present a generic approach to transform CSS codes by building upon their equivalence to phase-free ZX diagrams. Using the ZX calculus, we demonstrate diagrammatic transformations between encoding maps associated with different codes. As a motivating example, we give explicit transformations between the Steane code and the quantum Reed-Muller code, since by switching between these two codes, one can obtain a fault-tolerant universal gate set. To this end, we propose a bidirectional rewrite rule to find a (not necessarily transversal) physical implementation for any logical ZX diagram in any CSS code.

We then focus on two code transformation techniques: \textit{code morphing}, a procedure that transforms a code while retaining its fault-tolerant gates, and \textit{gauge fixing}, where complimentary codes can be obtained from a common subsystem code (e.g., the Steane and the quantum Reed-Muller codes from the $\fifteensub$ code). We provide explicit graphical derivations for these techniques and show how ZX and graphical encoder maps relate several equivalent perspectives on these code transforming operations.

\end{abstract}

 \section{Introduction}
Quantum computation has demonstrated its potential in speeding up large-scale computational tasks~\cite{arute2019quantum,zhu2022quantum} and revolutionizing multidisciplinary fields such as drug discovery~\cite{cao2018potential}, climate prediction~\cite{tennie2022quantum}, chemistry simulation~\cite{leontica2021simulating}, and the quantum internet~\cite{fang2022quantum}. However, in a quantum system, qubits are sensitive to interference and information becomes degraded~\cite{nielsen2001quantum}. To this end, quantum error correction~\cite{shor1995nineqcode,steane1996multiple} and fault tolerance~\cite{GottesmanD1998ftqc,knill1997theory} have been developed to achieve large-scale universal quantum computation~\cite{gottesman2022opportunities}. 

Stabilizer theory~\cite{gottesman1997stabilizer} is a mathematical framework to describe and analyze properties of quantum error-correcting codes (QECC). It is based on the concept of stabilizer groups, which are groups of Pauli operators whose joint $+1$ eigenspace corresponds to the code space. Stabilizer codes are a specific type of QECC whose encoder can be efficiently simulated \cite{aaronson2004improved,Gottesman1998Heisenberg}. As a family of stabilizer codes, Calderbank-Shor-Steane (CSS) codes permit simple code constructions from classical codes \cite{calderbank1996good,calderbank1997quantum,steane1996multiple,steane1996simple}.

As a language for rigorous diagrammatic reasoning of quantum computation, the ZX calculus consists of ZX diagrams and a set of rewrite rules \cite{CoeckeB2011interacting,vandeWetering2020zxreview}.
It has been used to relate stabilizer theory to graphical normal forms: notably, efficient axiomatization of the stabilizer fragments for qubits~\cite{backensZXcalculusCompleteStabilizer2014,HuAT2022graphzx,mcelvanney2022complete}, qutrits~\cite{townsendteague2022simplification,wangQutritZXcalculusComplete2018}, and prime-dimensional qudits~\cite{boothCompleteZXcalculiStabiliser2022}. This has enabled various applications,  such as measurement-based quantum computation~\cite{mcelvanney2022complete,SimmonsW2021pauliflow}, quantum circuit optimization~\cite{ColeO2022stabpauliexp,gorard2021zxcalculus} and verification~\cite{lehmann2022vyzx}, as well as classical simulation~\cite{CodsiJ2022stabdecomp,kissinger2021simulating}. Beyond these, ZX-calculus has been applied to verify QECC~\cite{DuncanR2014verifying713,GarvieL2017verifying832}, represent Clifford encoders~\cite{KhesinAB2023cliffordencoderszx}, as well as study various QECC such as tripartite coherent parity check codes~\cite{CaretteT2019SZX,ChancellorN2016CPCZX} and surface codes~\cite{GidneyC2022surfacecodepentagons,GidneyC2019msdCCZto2T,GidneyC2019surfacecodeautoccz, ShawATE2022surgerysurfacecodebus}. Specific to CSS codes, ZX-calculus has been used to visualize their encoders~\cite{KissingerA2022phasefreeCSS}, code maps and code surgeries~\cite{cowtan2023css}, their correspondence to affine Lagrangian relations~\cite{comfortAlgebraStabilizerCodes2023}, and their constructions in high-dimensional quantum systems~\cite{cowtan2022qudit}.

In this paper, we seek to answer some overarching questions about QECC constructions and fault-tolerant implementations. We focus on CSS codes and leverage the direct correspondence between phase-free ZX diagrams and CSS code encoders~\cite{KissingerA2022phasefreeCSS}. Given an arbitrary CSS code, based on its normal form, we propose a bidirectional rewrite rule to find a (not necessarily transversal) physical implementation for any logical ZX diagram. Furthermore, we demonstrate diagrammatic transformations between encoding maps associated with different codes. Here, we focus on two code transformation techniques: \textit{code morphing}, a procedure that transforms a code while retaining its fault-tolerant gates~\cite{VasmerM2022morphingcodes}, and \textit{gauge fixing}, where complimentary codes (such as the Steane and the quantum Reed-Muller codes) can be obtained from a common subsystem code~\cite{anderson2014fault,paetznick2013univftqctransversal,QuanDX2018gaugefixRMconvert,vuillot2019code}. We provide explicit graphical derivations for these techniques and show how ZX and graphical encoder maps relate several equivalent perspectives on these code transforming operations.

The rest of this paper is organized as follows. In \cref{sec:pre}, we introduce notions and techniques used to graphically transform different CSS codes using the ZX calculus. In \cref{sec:graphicalconstr}, we generalize the ZX normal form for CSS stabilizer codes to CSS subsystem codes, and provide generic bidirectional rewrite rules for any CSS encoder. In \cref{sec:morphing}, we provide explicit graphical derivations for morphing the Steane and the quantum Reed-Muller codes. In \cref{sec:codeswitching}, we focus on the switching protocol between these two codes. Through ZX calculus, we provide a graphical interpretation of this protocol as gauge-fixing the $\fifteensub$ subsystem code, followed by syndrome-determined recovery operations. We conclude with \cref{sec:conclude}.

 \section{Preliminaries}
\label{sec:pre}
We start with some definitions. The Pauli matrices are $2\times 2$ unitary operators acting on a single qubit. Let $i$ be the imaginary unit.
\[
I =\begin{bmatrix}
    1 & 0\\
    0 & 1
\end{bmatrix}, \quad  X=\begin{bmatrix}
    0 & 1\\
    1 & 0
\end{bmatrix}, \quad Z = \begin{bmatrix}
    1 & 0 \\
    0 & -1
\end{bmatrix}, \quad Y = iXZ = \begin{bmatrix}
    0 & -i \\
    i & 0
\end{bmatrix}.
\]

Let $\mathcal{P}_1$ be the single-qubit Pauli group, $\mathcal{P}_1 = \bigl \langle i, X, Z\bigr \rangle$, $I, Y \in \mathcal{P}_1$. 

\begin{definition}
    Let $U \in \mathcal{U}(2)$. In a system over $n$ qubits, $1 \leq i \leq n$, 
    \[
    U_i = I \otimes \ldots \otimes I \otimes U \otimes I \otimes \ldots \otimes I
    \]
    denotes $U$ acting on the $i$-th qubit, and identity on all other qubits.
\end{definition}

Let  $\mathcal{P}_n$ be the $n$-qubit Pauli group. It consists of all tensor products of single-qubit Pauli operators. 
    \[
\mathcal{P}_n = \bigl \langle i, X_1, Z_1, \ldots, X_n, Z_n\bigr \rangle.
\]

The stabilizer formalism is a mathematical framework to describe and analyze the properties of certain QECC, called stabilizer codes~\cite{gottesman1997stabilizer,GottesmanD1998ftqc}. Consider $n$ qubits and let $m \leq n$. A stabilizer group $\mathcal{S} = \bigl \langle S_1, \ldots, S_m\bigr \rangle$ is an Abelian subgroup of $\mathcal{P}_n$ that does not contain $-I$. The codespace of the corresponding stabilizer code, $\mathcal C$, is the joint $+1$ eigenspace of $\mathcal{S}$, i.e.,
\[
\mathcal{C} = \{\ket{\psi} \in \C^{2^n}; \ S \ket{\psi} = \ket{\psi}, \forall S \in \mathcal{S}\}.
\]

The number of encoded qubits in a stabilizer code is $k = n - m$, where $m$ is the number of independent stabilizer generators~\cite{gottesman1997stabilizer}. Moreover, we can define the \emph{centralizer} of $\mathcal{S}$ as
\[
\mathcal{N}(\mathcal{S})= \{U \in \mathcal{P}_n; \ [U,S] = 0, \forall S \in \
\mathcal{S}\}.
\]

One can check that $\mathcal{N}(\mathcal{S})$ is a subgroup of $\mathcal{P}_n$ and $\mathcal{S} \subset \mathcal{N}(\mathcal{S})$. 
We remark that the notions of normalizer and centralizer coincide for any stabilizer group. In what follows, we will use them interchangeably.
As we will see later, $\mathcal{N}(\mathcal{S})$ provides an algebraic structure for the subsystem codes. The code distance, $d$, of a stabilizer code is the minimal weight of operators in $\mathcal N(\mathcal S) / \langle iI \rangle$ that is not in $\mathcal S$.
We summarize the properties of a stabilizer code with the shorthand $\nkdcode$.

Finally, we introduce some notation for subsets of $n$-qubit Pauli operators, which will prove useful for defining CSS codes.

\begin{definition}
\label{def:notion}
Let $M$ be an $m\times n$ binary matrix and $P\in \mathcal{P}_1/\langle iI \rangle$. In the stabilizer formalism, M is called the \emph{stabilizer matrix}, and $M^P$ defines $m$ P-type stabilizer generators.
$$M^P\coloneqq \Biggl\{\bigotimes_{j=1}^n P^{[M]_{ij}}; \ 1\leq i\leq m\Biggr\}.$$ 
\end{definition}

CSS codes are QECC whose stabilizers are defined by two orthogonal binary matrices $G$ and $H$ \cite{calderbank1996good,steane1996multiple}:
\[
\mathcal{S}=\langle G^X, H^Z\rangle, \quad GH^\intercal=\mathbf{0},
\]
$H^\intercal$ is the transpose of $H$. This means that the stabilizer generators of a CSS code can be divided into two types: X-type and Z-type. For example, the $\steanecode$ Steane code~\cite{steane1996multiple} in \Cref{fig:steane} is specified by 
\begin{equation}
 G = H = \begin{bmatrix}
    1 & 0 & 1 & 0 & 1 & 0 & 1 \\
    0 & 1 & 1 & 0 & 0 & 1 & 1 \\
    0 & 0 & 0 & 1 & 1 & 1 & 1
  \end{bmatrix}_{3\times 7}.
  \label{eq:1}
\end{equation}  
Accordingly, the X-type and Z-type stabilizers are defined as 
\[
S^X_1 = X_1X_3X_5X_7, \; S^X_2 = X_2X_3X_6X_7, \; S^X_3 = X_4X_5X_6X_7, \; S^Z_1 = Z_1Z_3Z_5Z_7, \; S^Z_2 = Z_2Z_3Z_6Z_7, \; S^Z_3 = Z_4Z_5Z_6Z_7.
\]
The logical operators $\overline{X}$ and $\overline{Z}$ are defined as

\begin{equation}
\label{eq:logicalop}
\overline{X} = X_1X_4X_5 \qquad \text{and} \qquad \overline{Z} = Z_1Z_4Z_5.
\end{equation}

In \cref{subsec:subsysstemcodes}, we define CSS subsystem codes. In \cref{subsec:qrm}, we define several CSS codes that will be used in subsequent sections. In \cref{subsec:zx}, we introduce the basics of the ZX calculus and the phase-free ZX normal forms.

\subsection{CSS Subsystem Codes}
\label{subsec:subsysstemcodes}

Subsystem codes~\cite{kribs2005unified,poulin2005stabilizer} are QECC where some of the logical qubits are not used for information storage and processing. These logical qubits are called gauge qubits. By fixing gauge qubits to some specific states, the same subsystem code may exhibit different properties, for instance, having different sets of transversal gates~\cite{bombinGaugeColorCodes2015,kubicaUniversalTransversalGates2015,kubicaSingleshotQuantumError2022,paetznick2013univftqctransversal,yoderUniversalFaulttolerantQuantum2017}. This provides a tool to circumvent restrictions on transversal gates such as the Eastin-Knill theorem~\cite{eastin2009restrictions}. 

Based on the construction proposed in \cite{poulin2005stabilizer}, we describe a subsystem code using the stabilizer formalism. 

\begin{definition}
\label{def:subsystem}
    Given a stabilizer group $\mathcal{S}$, a gauge group $\mathcal{G}$ is a normal subgroup of $\mathcal{N}(\mathcal{S})$, such that $\mathcal{S} \subset \mathcal{G}$ and that $\mathcal{G}/\mathcal{S}$ contains anticommuting Pauli pairs. In other words, one can write 
    \[
    \mathcal{S} = \bigl \langle S_1, \ldots, S_m \bigr \rangle, \quad \mathcal{G} = \bigl \langle S_1, \ldots, S_m, g^X_1, g^Z_1, \ldots, g^X_r, g^Z_r\bigr \rangle, \quad 1 \leq m+r \leq n.
    \]
    $(\mathcal{S},\mathcal{G})$ defines an $\nkrdcode$ subsystem code where $n = m + k + r$. The logical operators are elements of the quotient group $\mathcal{L} = \mathcal{N}(\mathcal{S})/\mathcal{G}$.
\end{definition}

Under this construction, $n$ physical qubits are used to encode $k$ logical qubits with $r$ gauge qubits. Alternatively, we can think of the gauge group $\mathcal{G}$ as partitioning the code space $\mathcal{C}$ into two subsystems: $\mathcal{C} = \mathcal{A} \otimes \mathcal{B}$. Logical information is encoded in $\mathcal{A}$ and $\mathcal{L}$ serves as the group of logical operations. Gauge operators from $\mathcal{G}$ act trivially on subsystem $\mathcal{A}$, while operators from $\mathcal{L}$ act trivially on subsystem $\mathcal{B}$. 
Therefore, two states $\rho^\mathcal{A} \otimes \rho^\mathcal{B}$ and ${\rho'}^{\mathcal{A}} \otimes {\rho'}^{\mathcal{B}}$ are considered equivalent if $\rho^{\mathcal A} = {\rho'}^{\mathcal A}$, regardless of the states $\rho^\mathcal{B}$ and ${\rho'}^\mathcal{B}$. When $r = 0$, $\mathcal{G}=\mathcal{S}$. 
In that case, an $\nkzerodcode$ subsystem code is essentially an $\nkdcode$ stabilizer code.

CSS subsystem codes are subsystem codes whose stabilizer generators can be divided into X-type and Z-type operators. In what follows, we provide an example to illustrate their construction.

\subsection{Some Interesting CSS Codes}
\label{subsec:qrm}
We start by defining the stabilizer groups for the $\steanecode$ Steane code, the $\fifteencode$ extended Steane code~\cite{anderson2014fault}, and the $\fifteencode$ quantum Reed-Muller code~\cite{knill1996threshold}. They are derived from the family of $\qrmm$ quantum Reed-Muller codes, with a recursive construction of stabilizer matrices~\cite{steane1999quantum}. The Steane code has transversal logical Clifford operators, and the quantum Reed-Muller code has a transversal logical T gate. Together these operators form a universal set of fault-tolerant gates. In \Cref{sec:codeswitching}, the relations between these codes are studied from a diagrammatic perspective.

For brevity, their corresponding stabilizer groups are denoted as $\mathcal{S}_{steane}$, $\mathcal{S}_{ex}$, and $\mathcal{S}_{qrm}$. As per \Cref{def:notion}, consider three stabilizer matrices $F$, $H$, and $J$. Note that $G$ is defined in \Cref{eq:1}. $\mathbf{0}$ and $\mathbf{1}$ denote blocks of $0$s' and $1$s' respectively. Their dimensions can be inferred from the context. 

\begin{gather*}
  F = \left[ {\begin{array}{ccc}
    G & 0 & G  \\
    \mathbf{0} & 1 & \mathbf{1}
  \end{array} } \right]_{4\times 15},\quad
  H = \left[ {\begin{array}{cc}
    G & \mathbf{0}
  \end{array} } \right]_{3\times 15}, \\
  \vspace{1 em}
  J = \left[ {\begin{array}{ccccccccccccccc}
    1 & 0 & 1 & 0 & 0 & 0 & 0 & 0 & 1 & 0 & 1 & 0 & 0 & 0 & 0 \\
    0 & 1 & 1 & 0 & 0 & 0 & 0 & 0 & 0 & 1 & 1 & 0 & 0 & 0 & 0 \\
    0 & 0 & 1 & 0 & 0 & 0 & 1 & 0 & 0 & 0 & 1 & 0 & 0 & 0 & 1
  \end{array} } \right]_{3\times 15}.
\end{gather*}
Then, the stabilizer groups are defined as
\begin{equation}
\label{eq:stab-qrm}
    \mathcal{S}_{steane} = \bigl \langle G^X, G^Z \bigr\rangle, \quad \mathcal{S}_{ex} = \bigl \langle F^X, F^Z, H^X, H^Z \bigr\rangle, \quad \mathcal{S}_{qrm} = \bigl \langle F^X, F^Z, H^Z, J^Z \bigr\rangle.
\end{equation}

Geometrically, one can define $\mathcal{S}_{steane}$ and $\mathcal{S}_{qrm}$ with the aid of \Cref{fig:srm}. In \Cref{fig:steane}, the Steane code is visualized on a 2D lattice. Since the Steane code is self-dual, every coloured face corresponds to an X-type and Z-type stabilizer. In \Cref{fig:reed-muller}, the quantum Reed-Muller code is visualized on a 3D lattice. Every coloured face corresponds to a weight-$4$ Z-type stabilizer. Every coloured cell corresponds to a weight-$8$ X-type and Z-type stabilizer respectively. For the Steane code, the logical operators defined in \Cref{eq:logicalop} correspond to an edge in the triangle. For the quantum Reed-Muller code, the logical X operator corresponds to a weight-$7$ triangular face, and the logical Z operator corresponds to a weight-$3$ edge of the entire tetrahedron. An example is shown below.
\begin{equation}
    \label{eq:logicalqrm}
    \overline{X} = X_1 X_2 X_3 X_4 X_5 X_6 X_7 \qquad \text{and} \qquad \overline{Z} = Z_1 Z_4 Z_5
\end{equation}

Given such representations, the Steane code and the quantum Reed-Muller code are also special cases of colour codes \cite{bombin2006distillation,bombin2007nobraiding,kubicaUniversalTransversalGates2015}.

\begin{figure}[!ht]
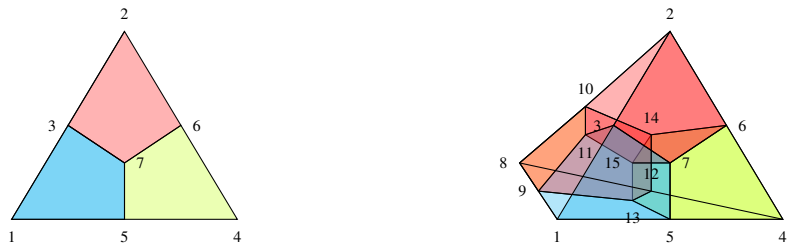

    \centering
    \begin{subfigure}{0.35\textwidth}
         \ctikzfig{figs/conv/steane}
         \caption{The Steane code as a 2D colour code.}
         \label{fig:steane}
     \end{subfigure}
    \begin{subfigure}{0.5\textwidth}
         \ctikzfig{figs/conv/qrm}
         \caption{The quantum Reed-Muller code as a 3D colour code.}
         \label{fig:reed-muller}
     \end{subfigure}
     \caption{Each vertex represents a physical qubit. Each edge serves as an aid to the eye. They do not imply any physical interactions or inherent structures.}
     \label{fig:srm}
\end{figure}

From \Cref{eq:stab-qrm}, the extended Steane code is self-dual, and its encoded state is characterized by the lemma below. It shows that $\mathcal{S}_{ex}$ and $\mathcal{S}_{steane}$ are equivalent up to some auxiliary state.

\begin{lemma}[\cite{anderson2014fault}]
\label{lem:encodedstate}
Any codeword $\ket{\psi}$ of the extended Steane code can be decomposed into a codeword $\ket{\phi}$ of the Steane code and a fixed state $\ket{\eta}$. That is,

\[
\ket{\psi}=\ket{\phi}\otimes\ket{\eta},
\]
where $\ket{\eta}=\frac{1}{\sqrt{2}}(\ket{0}\ket{\overline{0}}+\ket{1}\ket{\overline{1}})$, $\ket{\overline{0}}$ and $\ket{\overline{1}}$ are the logical 0 and 1 encoded in the Steane code.
\end{lemma}

Since the logical information $\ket{\phi}$ encoded in the Steane code is not entangled with $\ket{\eta}$, to switch between the Steane code and the extended Steane code, one may simply add or discard the auxiliary state $\ket{\eta}$. This property will prove useful in \Cref{sec:codeswitching}.

Next, we define the $\fifteensub$ CSS subsystem code~\cite{vuillot2019code}. As per \Cref{def:subsystem}, let $\mathcal{S}_{sub}$ and $\mathcal{G}$ be its stabilizer group and gauge group respectively.
\begin{equation}
\label{eq:stab-sub}
     \mathcal{S}_{sub} = \bigl \langle F^X, F^Z, H^Z \bigr\rangle, \quad \mathcal{G} = \bigl\langle F^X, F^Z, H^X, H^Z, J^Z \bigr\rangle.
\end{equation}

Let $\mathcal{L}_g = \mathcal{G}/\mathcal{S}$ and $\mathcal{L} = \mathcal{N}(\mathcal{S})/\mathcal{G}$. One can verify that 
\begin{equation}
    \mathcal{L}_g =  \bigl \langle H^X, J^Z \bigr\rangle, \quad \mathcal{L} = \bigl \langle \overline{X}, \overline{Z} \bigr\rangle.
\end{equation}
Thus, the CSS subsystem code has one logical qubit and three gauge qubits, and they are acted on by $\mathcal{L}$ and $\mathcal{L}_g$ respectively. From \Cref{sec:graphicalconstr} onwards, we call operators in $\mathcal{L}_g$ as \emph{gauge operators}.

Moreover, $\mathcal{S}_{sub}$ can be viewed as the stabilizer group of a $\intcode$ CSS code, with logical operators $\mathcal{L}'$. This code appears in an intermediary step of the gauge fixing process in \cref{sec:codeswitching}.
\begin{equation}
\mathcal{L}' \coloneqq\mathcal{L}_g\cup\mathcal{L} = \bigl \langle H^X, J^Z, \overline{X}, \overline{Z}\bigr\rangle.
\end{equation}

\subsection{ZX Calculus}
\label{subsec:zx}

The qubit ZX-calculus~\cite{CoeckeB2008ZX,CoeckeB2011interacting,PQP,vandeWetering2020zxreview} is a quantum graphical calculus for diagrammatic reasoning of any qubit quantum computation.  Every diagram in the calculus is composed of two types of generators: Z spiders, which sum over the eigenbasis of the Pauli Z operator:
\begin{equation}
    \begin{tikzpicture}
	\begin{pgfonlayer}{nodelayer}
		\node [style={green small (text)}] (0) at (0, 0) {$\alpha$};
		\node [style=none] (1) at (1, 0.5) {};
		\node [style=none] (2) at (-1, 0.5) {};
		\node [style=none] (3) at (-1, -0.5) {};
		\node [style=none] (4) at (1, -0.5) {};
		\node [style=none, rotate=90] (7) at (-0.75, 0) {$...$};
		\node [style=none, rotate=90] (8) at (0.75, 0) {$...$};
		\node [style=none] (9) at (-1.25, 0.5) {};
		\node [style=none] (10) at (-1.25, -0.5) {};
		\node [style=none] (11) at (-1.5, 0) {};
		\node [style=none] (12) at (-1.75, 0) {\tiny$m$};
		\node [style=none] (13) at (1.25, -0.5) {};
		\node [style=none] (14) at (1.25, 0.5) {};
		\node [style=none] (15) at (1.5, 0) {};
		\node [style=none] (16) at (1.75, 0) {\tiny$n$};
	\end{pgfonlayer}
	\begin{pgfonlayer}{edgelayer}
		\draw [style=none, in=-141, out=0, looseness=0.75] (3.center) to (0);
		\draw [style=none, in=180, out=-39, looseness=0.75] (0) to (4.center);
		\draw [style=none, in=180, out=39, looseness=0.75] (0) to (1.center);
		\draw [style=none, in=141, out=0, looseness=0.75] (2.center) to (0);
		\draw (9.center)
			 to [in=30, out=180, looseness=0.75] (11.center)
			 to [in=180, out=-30, looseness=0.75] (10.center);
		\draw (13.center)
			 to [in=-150, out=0, looseness=0.75] (15.center)
			 to [in=0, out=150, looseness=0.75] (14.center);
	\end{pgfonlayer}
\end{tikzpicture} \ \coloneqq \ \ket{0}^{\otimes n}\bra{0}^{\otimes m} \ +\  e^{i \alpha} \ket{1}^{\otimes n}\bra{1}^{\otimes m},\label{zspider}
\end{equation}
and X spiders, which sum over the eigenbasis of the Pauli X operator:
\begin{equation}
    \begin{tikzpicture}
	\begin{pgfonlayer}{nodelayer}
		\node [style={red small (text)}] (0) at (0, 0) {$\alpha$};
		\node [style=none] (1) at (1, 0.5) {};
		\node [style=none] (2) at (-1, 0.5) {};
		\node [style=none] (3) at (-1, -0.5) {};
		\node [style=none] (4) at (1, -0.5) {};
		\node [style=none, rotate=90] (7) at (-0.75, 0) {$...$};
		\node [style=none, rotate=90] (8) at (0.75, 0) {$...$};
		\node [style=none] (9) at (-1.25, 0.5) {};
		\node [style=none] (10) at (-1.25, -0.5) {};
		\node [style=none] (11) at (-1.5, 0) {};
		\node [style=none] (12) at (-1.75, 0) {\tiny$m$};
		\node [style=none] (13) at (1.25, -0.5) {};
		\node [style=none] (14) at (1.25, 0.5) {};
		\node [style=none] (15) at (1.5, 0) {};
		\node [style=none] (16) at (1.75, 0) {\tiny$n$};
	\end{pgfonlayer}
	\begin{pgfonlayer}{edgelayer}
		\draw [style=none, in=-141, out=0, looseness=0.75] (3.center) to (0);
		\draw [style=none, in=180, out=-39, looseness=0.75] (0) to (4.center);
		\draw [style=none, in=180, out=39, looseness=0.75] (0) to (1.center);
		\draw [style=none, in=141, out=0, looseness=0.75] (2.center) to (0);
		\draw (9.center)
			 to [in=30, out=180, looseness=0.75] (11.center)
			 to [in=180, out=-30, looseness=0.75] (10.center);
		\draw (13.center)
			 to [in=-150, out=0, looseness=0.75] (15.center)
			 to [in=0, out=150, looseness=0.75] (14.center);
	\end{pgfonlayer}
\end{tikzpicture} \ \coloneqq \ \ket{+}^{\otimes n}\bra{+}^{\otimes m} \ +\  e^{i \alpha} \ket{-}^{\otimes n}\bra{-}^{\otimes m}.\label{xspider}
\end{equation}

The ZX-calculus is \emph{universal}~\cite{CoeckeB2011interacting} in the sense that any linear map from $m$
qubits to $n$ qubits corresponds exactly to a ZX diagram, by the construction of \cref{zspider,xspider} and the composition of linear maps.

Furthermore, the ZX-calculus is \emph{complete}~\cite{HadzihasanovicA2018zxcomplete,JeandelE2020completeZX}: Any equality of linear maps on any number of qubits derivable in the Hilbert space formalism, is derivable using only a finite set of rules in the calculus. The smallest complete rule set to date~\cite{VilmartR2019nearminzx} is shown in~\cref{fig:zxrules}. Some additional rules, despite being derivable from this rule set, will be convenient to use in this paper. They are summarized in \cref{fig:zxadditionalrules}.

\begin{figure}
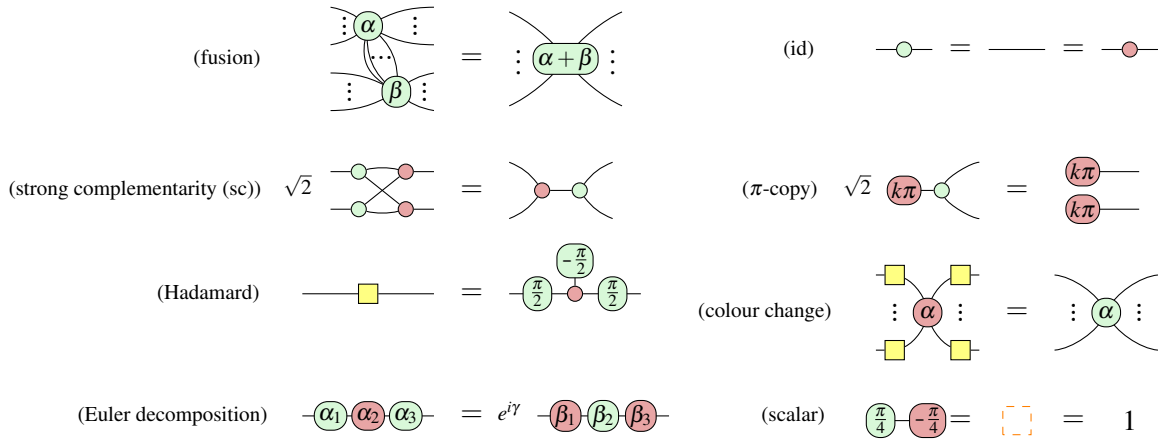

    \centering
    \ctikzfig{figs/prelim/zxrules}
    \caption{These eight equations suffice to derive all other equalities of linear maps on qubits~\cite{VilmartR2019nearminzx}. $k \in \Z_2$. $\alpha_i$, $\beta_i$ and $\gamma$ are real numbers satisfying the trigonometric relations derived in~\cite{CoeckeWang2018deriveeuler}. Each equation still holds when we replace all spiders with their corresponding spiders of the opposite colour. Whenever there are any two wires with $...$ between them, the rule holds when replacing this with any number of wires (i.e., $0$ or greater).}
    \label{fig:zxrules}
\end{figure}

When a spider has phase zero, we omit its phase in the diagram, as shown below. A ZX diagram is \textit{phase-free} if all of its spiders have zero phases. For more discussions on phase-free ZX diagrams, we refer readers to consult \cite{KissingerA2022phasefreeCSS}.
\[
    \scalebox{.9}{\begin{tikzpicture}
	\begin{pgfonlayer}{nodelayer}
		\node [style=red] (1) at (4, 2) {};
		\node [style=green] (4) at (-7, 2) {};
		\node [style=none] (5) at (-7.75, 2.75) {};
		\node [style=none] (6) at (-7.75, 1.25) {};
		\node [style=none] (7) at (-6.25, 2.75) {};
		\node [style=none] (8) at (-6.25, 1.25) {};
		\node [style=none] (9) at (-7.75, 2.25) {$\vdots$};
		\node [style=none] (14) at (-2.25, 1.25) {};
		\node [style=none] (15) at (-2.25, 2.75) {};
		\node [style=none] (17) at (-3.75, 2.75) {};
		\node [style=none] (18) at (-3.75, 1.25) {};
		\node [style=none] (19) at (3.25, 2.75) {};
		\node [style=none] (20) at (4.75, 2.75) {};
		\node [style=none] (21) at (3.25, 1.25) {};
		\node [style=none] (22) at (4.75, 1.25) {};
		\node [style={red (text)}] (24) at (8, 2) {};
		\node [style={red (text)}] (25) at (8, 2) {$0$};
		\node [style=none] (26) at (7.25, 2.75) {};
		\node [style=none] (27) at (8.75, 2.75) {};
		\node [style=none] (28) at (7.25, 1.25) {};
		\node [style=none] (29) at (8.75, 1.25) {};
		\node [style=none] (32) at (-5, 2) {$\coloneqq$};
		\node [style=none] (33) at (-6.25, 2.25) {$\vdots$};
		\node [style=none] (34) at (6, 2) {$\coloneqq$};
		\node [style=none] (35) at (-3.75, 2.25) {$\vdots$};
		\node [style=none] (36) at (-2.25, 2.25) {$\vdots$};
		\node [style=none] (37) at (3.25, 2.25) {$\vdots$};
		\node [style=none] (38) at (4.75, 2.25) {$\vdots$};
		\node [style=none] (39) at (7.25, 2.25) {$\vdots$};
		\node [style=none] (40) at (8.75, 2.25) {$\vdots$};
		\node [style={green (text)}] (41) at (-3, 2) {};
		\node [style={green (text)}] (42) at (-3, 2) {$0$};
	\end{pgfonlayer}
	\begin{pgfonlayer}{edgelayer}
		\draw [style=edge] (5.center) to (4);
		\draw [style=edge] (4) to (7.center);
		\draw [style=edge] (4) to (6.center);
		\draw [style=edge] (4) to (8.center);
		\draw [style=edge] (1) to (19.center);
		\draw [style=edge] (1) to (20.center);
		\draw [style=edge] (1) to (22.center);
		\draw [style=edge] (1) to (21.center);
		\draw [style=edge] (26.center) to (25);
		\draw [style=edge] (27.center) to (25);
		\draw [style=edge] (25) to (29.center);
		\draw [style=edge] (25) to (28.center);
		\draw (42) to (17.center);
		\draw (42) to (15.center);
		\draw (42) to (18.center);
		\draw (42) to (14.center);
	\end{pgfonlayer}
\end{tikzpicture}}
\]

\vspace{-2 em}

Due to the universality of the ZX calculus, quantum error-correcting code encoders, as linear isometries, can be drawn as ZX diagrams \cite{KhesinAB2023cliffordencoderszx}. Moreover, the encoder for a CSS code corresponds exactly to the phase-free \textit{ZX (and XZ) normal form}~\cite{KissingerA2022phasefreeCSS}. 
\begin{definition}
    \label{def:normalform}
    For a CSS stabilizer code defined by $\mathcal{S}$, let $\bigl\{\SX{i}; 1 \leq i \leq m\bigr\} \subset \mathcal{S}$ be the X-type stabilizer generators and $\bigl\{\overline{X_j}; 1 \leq j \leq k\bigr\}$ be the logical X operators, $m+k < n$. Its ZX normal form can be found via the following steps:
\begin{enumerate}[label={(\alph*)}]
        \item For each physical qubit, introduce an $X$ spider.
        \item For each X-type stabilizer generator $\SX{i}$ and logical operator $\overline{X_j}$, introduce a $Z$ spider and connect it to all X spiders where this operator has support.
        \item Give each $X$ spider an output wire.
        \item For each $Z$ spider representing $\overline{X_j}$, give it an input wire.
\end{enumerate}
\end{definition}

As an example, the ZX normal form for the Steane code is drawn in \cref{fig:qrm3encoder}. The XZ normal form can be constructed based on Z-type stabilizer generators and logical Z operators by inverting the roles of X and Z spiders in the above procedure. In \cite{KissingerA2022phasefreeCSS}, Kissinger gave an algorithm to rewrite any phase-free ZX diagram into both the ZX and XZ normal forms, and pointed out that it is sufficient to represent a CSS code encoder using either one of the forms.

 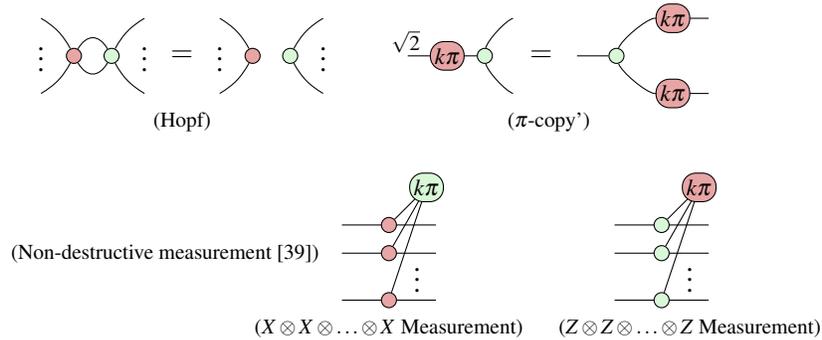
\begin{figure}[!ht]
    \centering
        \scalebox{1}{\begin{tikzpicture}
	\begin{pgfonlayer}{nodelayer}
		\node [style=none] (5) at (-11.25, 6.625) {$\vdots$};
		\node [style=none] (6) at (-8.5, 6.625) {$\vdots$};
		\node [style=none] (7) at (-11.25, 7.5) {};
		\node [style=none] (8) at (-11.25, 5.5) {};
		\node [style=none] (9) at (-8.5, 7.5) {};
		\node [style=none] (10) at (-8.5, 5.5) {};
		\node [style=red small] (11) at (-10.375, 6.5) {};
		\node [style=green small] (12) at (-9.375, 6.5) {};
		\node [style=none] (13) at (-7.5, 6.5) {$=$};
		\node [style=none] (14) at (-6.5, 6.625) {$\vdots$};
		\node [style=none] (15) at (-3.75, 6.625) {$\vdots$};
		\node [style=none] (16) at (-6.5, 7.5) {};
		\node [style=none] (17) at (-6.5, 5.5) {};
		\node [style=none] (18) at (-3.75, 7.5) {};
		\node [style=none] (19) at (-3.75, 5.5) {};
		\node [style=red small] (20) at (-5.625, 6.5) {};
		\node [style=green small] (21) at (-4.625, 6.5) {};
		\node [style=none] (22) at (-7.5, 4.75) {\scriptsize (Hopf)};
		\node [style=none] (23) at (-7.95, 1.25) {\scriptsize (Non-destructive measurement \cite{KissingerA2022phasefreeCSS})};
		\node [style=none] (24) at (-3.25, 2) {};
		\node [style=none] (25) at (-3.25, 1.25) {};
		\node [style=none] (26) at (-3.25, 0) {};
		\node [style=none] (27) at (-0.75, 2) {};
		\node [style=none] (28) at (-0.75, 1.25) {};
		\node [style=none] (29) at (-0.75, 0) {};
		\node [style=red small] (30) at (-2, 2) {};
		\node [style=red small] (31) at (-2, 1.25) {};
		\node [style=red small] (32) at (-2, 0) {};
		\node [style={green small (text)}] (33) at (-1, 3) {$k\pi$};
		\node [style=none] (34) at (4, 2) {};
		\node [style=none] (35) at (4, 1.25) {};
		\node [style=none] (36) at (4, 0) {};
		\node [style=none] (37) at (6.5, 2) {};
		\node [style=none] (38) at (6.5, 1.25) {};
		\node [style=none] (39) at (6.5, 0) {};
		\node [style=green small] (40) at (5.25, 2) {};
		\node [style=green small] (41) at (5.25, 1.25) {};
		\node [style=green small] (42) at (5.25, 0) {};
		\node [style={red small (text)}] (43) at (6.25, 3) {$k\pi$};
		\node [style=none] (44) at (1.3, 5.5) {};
		\node [style=green small] (45) at (0.55, 6.5) {};
		\node [style=none] (46) at (1.3, 7.5) {};
		\node [style={red small (text)}] (47) at (-0.45, 6.5) {$k\pi$};
		\node [style=none] (48) at (2.05, 6.5) {$=$};
		\node [style=none] (49) at (-1.55, 6.875) {\scriptsize $\sqrt{2}$};
		\node [style=none] (50) at (2.25, 4.75) {\scriptsize ($\pi$-copy')};
		\node [style=none] (51) at (-1.5, 6.5) {};
		\node [style=none] (52) at (5.05, 5.5) {};
		\node [style=green small] (53) at (4.05, 6.5) {};
		\node [style=none] (54) at (5.05, 7.5) {};
		\node [style=none] (57) at (3, 6.5) {};
		\node [style={red small (text)}] (58) at (5.55, 7.5) {$k\pi$};
		\node [style={red small (text)}] (59) at (5.55, 5.5) {$k\pi$};
		\node [style=none] (60) at (6.5, 7.5) {};
		\node [style=none] (61) at (6.5, 5.5) {};
		\node [style=none] (62) at (-2, -0.75) {\scriptsize ($X\otimes X\otimes  \ldots \otimes X$ Measurement)};
		\node [style=none] (63) at (6, -0.75) {\scriptsize ($Z\otimes Z\otimes \ldots \otimes Z$ Measurement)};
		\node [style=none] (65) at (-1.25, 0.75) {$\vdots$};
		\node [style=none] (66) at (6, 0.75) {$\vdots$};
	\end{pgfonlayer}
	\begin{pgfonlayer}{edgelayer}
		\draw [bend right=15] (9.center) to (12);
		\draw [bend left=15] (10.center) to (12);
		\draw [bend left=60, looseness=1.50] (12) to (11);
		\draw [bend left=15] (11) to (8.center);
		\draw [bend right=15] (11) to (7.center);
		\draw [bend right=15] (18.center) to (21);
		\draw [bend left=15] (19.center) to (21);
		\draw [bend left=15] (20) to (17.center);
		\draw [bend right=15] (20) to (16.center);
		\draw [bend left=60, looseness=1.50] (11) to (12);
		\draw (24.center) to (27.center);
		\draw (25.center) to (28.center);
		\draw (26.center) to (29.center);
		\draw (30) to (33);
		\draw (33) to (31);
		\draw (32) to (33);
		\draw (34.center) to (37.center);
		\draw (35.center) to (38.center);
		\draw (36.center) to (39.center);
		\draw (40) to (43);
		\draw (43) to (41);
		\draw (42) to (43);
		\draw [bend right, looseness=0.50] (46.center) to (45);
		\draw [bend right, looseness=0.50] (45) to (44.center);
		\draw (45) to (47);
		\draw (47) to (51.center);
		\draw [bend right, looseness=0.50] (54.center) to (53);
		\draw [bend right, looseness=0.50] (53) to (52.center);
		\draw (60.center) to (54.center);
		\draw (61.center) to (52.center);
		\draw (53) to (57.center);
	\end{pgfonlayer}
\end{tikzpicture}}
        \caption{Some other useful rewrite rules, each derivable from the rules in Figure~\ref{fig:zxrules}. $k \in \Z_2$. Each equation still holds when we interchange X and Z spiders.}
        \label{fig:zxadditionalrules}
\end{figure}

 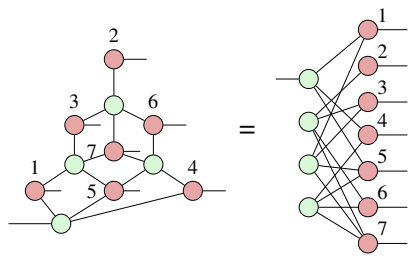
\begin{figure}[!ht]
        \centering
        \scalebox{.7}{\begin{tikzpicture}
	\begin{pgfonlayer}{nodelayer}
		\node [style=red, label={above:1}] (1) at (-3, -2.25) {};
		\node [style=red, label={above: 2}] (2) at (0, 2.75) {};
		\node [style=red, label={above: 3}] (3) at (-1.5, 0.25) {};
		\node [style=red, label={above: 4}] (4) at (3, -2.25) {};
		\node [style=red, label={left: 5}] (5) at (0, -2.25) {};
		\node [style=red, label={above: 6}] (6) at (1.5, 0.25) {};
		\node [style=red, label={left: 7}] (7) at (0, -0.75) {};
		\node [style=green] (8) at (0, 1) {};
		\node [style=green] (9) at (-1.5, -1.25) {};
		\node [style=green] (10) at (1.5, -1.25) {};
		\node [style=green] (11) at (-2, -3.5) {};
		\node [style=none] (12) at (-4, -3.5) {};
		\node [style=none] (13) at (1.25, 2.75) {};
		\node [style=none] (14) at (2.75, 0.25) {};
		\node [style=none] (15) at (4.25, -2.25) {};
		\node [style=none] (16) at (1, -2.25) {};
		\node [style=none] (17) at (-2, -2.25) {};
		\node [style=none] (18) at (1, -0.75) {};
		\node [style=none] (19) at (-0.5, 0.25) {};
	\end{pgfonlayer}
	\begin{pgfonlayer}{edgelayer}
		\draw [style=edge] (2) to (8);
		\draw [style=edge] (8) to (6);
		\draw [style=edge] (8) to (7);
		\draw [style=edge] (8) to (3);
		\draw [style=edge] (3) to (9);
		\draw [style=edge] (9) to (7);
		\draw [style=edge] (9) to (5);
		\draw [style=edge] (9) to (1);
		\draw [style=edge] (10) to (7);
		\draw [style=edge] (10) to (6);
		\draw [style=edge] (10) to (5);
		\draw [style=edge] (10) to (4);
		\draw [style=edge] (1) to (11);
		\draw [style=edge] (11) to (5);
		\draw [style=edge] (11) to (4);
		\draw [style=edge] (12.center) to (11);
		\draw (2) to (13.center);
		\draw (6) to (14.center);
		\draw (4) to (15.center);
		\draw (5) to (16.center);
		\draw (7) to (18.center);
		\draw (1) to (17.center);
		\draw (3) to (19.center);
	\end{pgfonlayer}
\end{tikzpicture}}\ = \ \scalebox{.7}{\begin{tikzpicture}[tikzfig]
	\begin{pgfonlayer}{nodelayer}
		\node [style=red, label={[label distance=-4]20: 1}] (1) at (0.75, 7.75) {};
		\node [style=red, label={[label distance=-4]20: 2}] (2) at (0.75, 5) {};
		\node [style=red, label={[label distance=-4]20: 3}] (3) at (0.75, 2.25) {};
		\node [style=red, label={[label distance=-4]20: 4}] (4) at (0.75, -0.25) {};
		\node [style=red, label={[label distance=-4]20: 5}] (5) at (0.75, -3) {};
		\node [style=red, label={[label distance=-4]20: 6}] (6) at (0.75, -5.75) {};
		\node [style=red, label={[label distance=-4]20: 7}] (7) at (0.75, -8.5) {};
		\node [style=green] (8) at (-3.75, 0.75) {};
		\node [style=green] (9) at (-3.75, -2.5) {};
		\node [style=green] (10) at (-3.75, -5.75) {};
		\node [style=green] (11) at (-3.75, 4) {};
		\node [style=none] (12) at (-6.25, 4) {};
		\node [style=none] (13) at (3.75, 5) {};
		\node [style=none] (14) at (3.75, -5.75) {};
		\node [style=none] (15) at (3.75, -0.25) {};
		\node [style=none] (16) at (3.75, -3) {};
		\node [style=none] (17) at (3.75, 7.75) {};
		\node [style=none] (18) at (3.75, -8.5) {};
		\node [style=none] (19) at (3.75, 2.25) {};
	\end{pgfonlayer}
	\begin{pgfonlayer}{edgelayer}
		\draw [style=edge] (2) to (8);
		\draw [style=edge] (8) to (6);
		\draw [style=edge] (8) to (7);
		\draw [style=edge] (8) to (3);
		\draw [style=edge] (3) to (9);
		\draw [style=edge] (9) to (7);
		\draw [style=edge] (9) to (5);
		\draw [style=edge] (9) to (1);
		\draw [style=edge] (10) to (7);
		\draw [style=edge] (10) to (6);
		\draw [style=edge] (10) to (5);
		\draw [style=edge] (10) to (4);
		\draw [style=edge] (1) to (11);
		\draw [style=edge] (11) to (5);
		\draw [style=edge] (11) to (4);
		\draw [style=edge] (12.center) to (11);
		\draw (2) to (13.center);
		\draw (6) to (14.center);
		\draw (4) to (15.center);
		\draw (5) to (16.center);
		\draw (7) to (18.center);
		\draw (1) to (17.center);
		\draw (3) to (19.center);
	\end{pgfonlayer}
\end{tikzpicture}} 
        \caption{The Steane code encoder in the ZX normal form.}
        \label{fig:qrm3encoder}
\end{figure}

  \section{Graphical Construction of CSS Encoders}
\label{sec:graphicalconstr}
\subsection{ZX Normal Forms for CSS Subsystem Codes}
\label{subsec:zxcsssub}
We generalize the ZX normal form for CSS stabilizer codes to CSS subsystem codes as follows.

\begin{definition}
    \label{def:subsystemnormalform}
    For an $\nkrdcode$ CSS subsystem code defined by $(\mathcal{S},\mathcal{G})$, let $\bigl\{\SX{i}; 1 \leq i \leq m\bigr\}$ be the X-type stabilizer generators, $\bigl\{\LGX{t}; 1 \leq t \leq r\bigr\}$ be the X-type gauge operators, and $\bigl\{\overline{X_j}; 1 \leq j \leq k\bigr\}$ be the logical X operators, $m+k+r < n$. Its ZX normal form can be found via the following steps:

    \begin{enumerate}[label={(\alph*)}]
        \item For each physical qubit, introduce an $X$ spider.
        \item For each stabilizer generator $\SX{i}$, logical operator $\overline{X_j}$ and gauge operator $\LGX{t}$, introduce a $Z$ spider and connect it to all X spiders where this operator has support.
        \item Give each $X$ spider an output wire.
        \item For each $Z$ spider representing $\overline{X_j}$, give it an input wire.
        \item For all $Z$ spiders representing $\LGX{t}$, attach to them a joint arbitrary input state (i.e., a density operator $\rho$).
    \end{enumerate}

\end{definition}
Similar to CSS stabilizer codes, CSS subsystem codes also have an equivalent XZ normal form, which can be found by inverting the role of Z and X in the above procedure.

For $n > 3$, below we exemplify the ZX normal form for an $\examplesub$ CSS subsystem code with three X-type stabilizers generators $\bigl\{\SX{1},\SX{2},\SX{3}\bigr\}$, two X-type gauge operators $\bigl\{\LGX{1},\LGX{2}\bigr\}$, and one logical operator $\bigl\{\overline{X}\bigr\}$. For simplicity, we substitute wires connecting Z and X spiders by $\begin{tikzpicture}
	\begin{pgfonlayer}{nodelayer}
		\node [style=none] (129) at (0, 0.5) {};
		\node [style=none] (130) at (0, -0.5) {};
		\node [style=none] (131) at (0.5, 0.15) {$\vdots$};
		\node [style=none] (132) at (-0.5, 0.25) {};
		\node [style=none] (133) at (-0.5, -0.25) {};
	\end{pgfonlayer}
	\begin{pgfonlayer}{edgelayer}
		\draw (132.center) to (129.center);
		\draw (133.center) to (130.center);
	\end{pgfonlayer}
\end{tikzpicture}
$. The detailed connectivities are omitted here, but they should be clear following step (b) in \cref{def:subsystemnormalform}. This notation will be used in the remainder of this paper.
$$\scalebox{0.8}
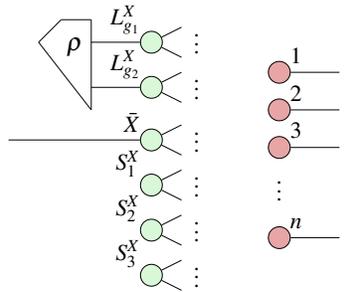
{\begin{tikzpicture}
	\begin{pgfonlayer}{nodelayer}
		\node [style=red, label={[label distance=-4]20: 1}] (4) at (0, 2.25) {};
		\node [style=green, label={[label distance=-4]175:$L_{g_2}^X$}] (35) at (-4.25, 1.75) {};
		\node [style=green, label={[label distance=-4]175:$L_{g_1}^X$}] (65) at (-4.25, 3.25) {};
		\node [style=none] (73) at (2, 2.25) {};
		\node [style=red, label={[label distance=-4]20: 2}] (74) at (0, 1) {};
		\node [style=none] (75) at (2, 1) {};
		\node [style=red, label={[label distance=-4]20: 3}] (76) at (0, -0.25) {};
		\node [style=none] (77) at (2, -0.25) {};
		\node [style=red, label={[label distance=-4]20: $n$}] (88) at (0, -3.25) {};
		\node [style=none] (89) at (2, -3.25) {};
		\node [style=none] (90) at (0, -1.5) {$\vdots$};
		\node [style=green, label={[label distance=-4]175:$S_2^X$}] (109) at (-4.25, -3) {};
		\node [style=none] (110) at (-3.25, -2.5) {};
		\node [style=none] (111) at (-3.25, -3.5) {};
		\node [style=none] (112) at (-2.75, -2.85) {$\vdots$};
		\node [style=green, label={[label distance=-4]175:$\bar{X}$}] (117) at (-4.25, 0) {};
		\node [style=none] (118) at (-3.25, 0.5) {};
		\node [style=none] (119) at (-3.25, -0.5) {};
		\node [style=none] (120) at (-2.75, 0.15) {$\vdots$};
		\node [style=green, label={[label distance=-4]175:$S_1^X$}] (121) at (-4.25, -1.5) {};
		\node [style=none] (122) at (-3.25, -1) {};
		\node [style=none] (123) at (-3.25, -2) {};
		\node [style=none] (124) at (-2.75, -1.35) {$\vdots$};
		\node [style=none] (125) at (-9, 0) {};
		\node [style=none] (126) at (-3.25, 2.25) {};
		\node [style=none] (127) at (-3.25, 1.25) {};
		\node [style=none] (128) at (-2.75, 1.9) {$\vdots$};
		\node [style=none] (129) at (-3.25, 3.75) {};
		\node [style=none] (130) at (-3.25, 2.75) {};
		\node [style=none] (131) at (-2.75, 3.4) {$\vdots$};
		\node [style=none] (134) at (-6.25, 3.25) {};
		\node [style=none] (135) at (-6.25, 1.75) {};
		\node [style=none] (136) at (-6.25, 4) {};
		\node [style=none] (137) at (-6.25, 1) {};
		\node [style=none] (138) at (-8, 3.25) {};
		\node [style=none] (139) at (-7.25, 4) {};
		\node [style=none] (140) at (-6.25, 1) {};
		\node [style=none] (141) at (-6.8, 3.1) {\large $\rho$};
		\node [style=green, label={[label distance=-4]175:$S_3^X$}] (142) at (-4.25, -4.5) {};
		\node [style=none] (143) at (-3.25, -4) {};
		\node [style=none] (144) at (-3.25, -5) {};
		\node [style=none] (145) at (-2.75, -4.35) {$\vdots$};
	\end{pgfonlayer}
	\begin{pgfonlayer}{edgelayer}
		\draw (4) to (73.center);
		\draw (74) to (75.center);
		\draw (76) to (77.center);
		\draw (88) to (89.center);
		\draw (109) to (110.center);
		\draw (109) to (111.center);
		\draw (117) to (118.center);
		\draw (117) to (119.center);
		\draw (121) to (122.center);
		\draw (121) to (123.center);
		\draw (125.center) to (117);
		\draw (65) to (129.center);
		\draw (65) to (130.center);
		\draw (35) to (126.center);
		\draw (35) to (127.center);
		\draw (136.center) to (140.center);
		\draw (140.center) to (138.center);
		\draw (138.center) to (139.center);
		\draw (139.center) to (136.center);
		\draw (134.center) to (65);
		\draw (135.center) to (35);
		\draw (142) to (143.center);
		\draw (142) to (144.center);
	\end{pgfonlayer}
\end{tikzpicture}}.$$

\subsection{Pushing through the Encoder}
For any $\nkdcode$ CSS code, its encoder map $E$ is of the form:
$$k \big\{ \ \begin{tikzpicture}
	\begin{pgfonlayer}{nodelayer}
		\node [style=rec] (0) at (-0.25, 0) {$E$};
		\node [style=none] (1) at (0.5, 0.75) {};
		\node [style=none] (3) at (0.5, -0.75) {};
		\node [style=none] (4) at (2.25, 0.75) {};
		\node [style=none] (6) at (2.25, -0.75) {};
		\node [style=none] (7) at (-1, 0.5) {};
		\node [style=none] (8) at (-2.25, 0.5) {};
		\node [style=none] (9) at (-1, -0.5) {};
		\node [style=none] (10) at (-2.25, -0.5) {};
		\node [style=none] (11) at (-1.75, 0.25) {$\vdots$};
		\node [style=none] (12) at (1.5, 0.25) {$\vdots$};
	\end{pgfonlayer}
	\begin{pgfonlayer}{edgelayer}
		\draw [style=edge] (8.center) to (7.center);
		\draw [style=edge] (10.center) to (9.center);
		\draw [style=edge] (1.center) to (4.center);
		\draw [style=edge] (3.center) to (6.center);
	\end{pgfonlayer}
\end{tikzpicture}
 \ \bigg\} n.$$

\begin{definition}
\label{def:logphyop}
    Let $\overline{X_i}$ and $\overline{Z_i}$ be the X and Z operators acting on the $i$-th logical qubit. Let $\overline{\mathcal{X}_i}$ and $\overline{\mathcal{Z}_i}$ be the physical implementation of $\overline{X_i}$ and $\overline{Z_i}$ respectively. Diagrammatically, they can be represented as $$\begin{tikzpicture}
	\begin{pgfonlayer}{nodelayer}
		\node [style=none] (84) at (-9.5, 4.75) {$=$};
		\node [style=none] (132) at (-12.625, 6) {};
		\node [style=none] (133) at (-11.125, 6) {};
		\node [style=none] (134) at (-11.125, 3.25) {};
		\node [style=none] (135) at (-12.625, 3.25) {};
		\node [style=none] (136) at (-11.125, 5.75) {};
		\node [style=none] (139) at (-11.125, 3.5) {};
		\node [style=none] (140) at (-12.625, 4) {};
		\node [style=none] (141) at (-12.625, 5.25) {};
		\node [style=none] (142) at (-11.875, 4.5) {$E$};
		\node [style=none] (145) at (-14.5, 6) {};
		\node [style=none] (146) at (-14.5, 3.25) {};
		\node [style=none] (147) at (-13.5, 3.25) {};
		\node [style=none] (148) at (-13.5, 6) {};
		\node [style=none] (149) at (-15, 5.25) {};
		\node [style=none] (150) at (-15, 4) {};
		\node [style=none] (151) at (-14.5, 5.25) {};
		\node [style=none] (152) at (-14.5, 4) {};
		\node [style=none] (153) at (-13.5, 5.25) {};
		\node [style=none] (154) at (-13.5, 4) {};
		\node [style=none] (155) at (-14, 4.5) {};
		\node [style=none] (156) at (-14, 4.5) {$\overline{X_1}$};
		\node [style=none] (178) at (-10.25, 3.5) {};
		\node [style=none] (181) at (-10.25, 5.75) {};
		\node [style=none] (185) at (-8, 6) {};
		\node [style=none] (186) at (-6.5, 6) {};
		\node [style=none] (187) at (-6.5, 3.25) {};
		\node [style=none] (188) at (-8, 3.25) {};
		\node [style=none] (189) at (-6.5, 5.75) {};
		\node [style=none] (192) at (-6.5, 3.5) {};
		\node [style=none] (193) at (-8, 4) {};
		\node [style=none] (194) at (-8, 5.25) {};
		\node [style=none] (195) at (-7.25, 4.5) {$E$};
		\node [style=none] (200) at (-8.75, 5.25) {};
		\node [style=none] (201) at (-8.75, 4) {};
		\node [style=none] (202) at (-5.75, 6) {};
		\node [style=none] (203) at (-4.5, 6) {};
		\node [style=none] (204) at (-4.5, 3.25) {};
		\node [style=none] (205) at (-5.75, 3.25) {};
		\node [style=none] (206) at (-4.5, 5.75) {};
		\node [style=none] (209) at (-4.5, 3.5) {};
		\node [style=none] (211) at (-5.75, 5.25) {};
		\node [style=none] (212) at (-5, 4.5) {$\overline{\mathcal{X}_1}$};
		\node [style=none] (213) at (-5.75, 5.75) {};
		\node [style=none] (215) at (-5.75, 3.5) {};
		\node [style=none] (216) at (-4, 5.75) {};
		\node [style=none] (219) at (-4, 3.5) {};
		\node [style=none] (220) at (4.75, 4.75) {$=$};
		\node [style=none] (221) at (1.625, 6) {};
		\node [style=none] (222) at (3.125, 6) {};
		\node [style=none] (223) at (3.125, 3.25) {};
		\node [style=none] (224) at (1.625, 3.25) {};
		\node [style=none] (225) at (3.125, 5.75) {};
		\node [style=none] (228) at (3.125, 3.5) {};
		\node [style=none] (229) at (1.625, 4) {};
		\node [style=none] (230) at (1.625, 5.25) {};
		\node [style=none] (231) at (2.375, 4.5) {$E$};
		\node [style=none] (232) at (-0.25, 6) {};
		\node [style=none] (233) at (-0.25, 3.25) {};
		\node [style=none] (234) at (0.75, 3.25) {};
		\node [style=none] (235) at (0.75, 6) {};
		\node [style=none] (236) at (-0.75, 5.25) {};
		\node [style=none] (237) at (-0.75, 4) {};
		\node [style=none] (238) at (-0.25, 5.25) {};
		\node [style=none] (239) at (-0.25, 4) {};
		\node [style=none] (240) at (0.75, 5.25) {};
		\node [style=none] (241) at (0.75, 4) {};
		\node [style=none] (242) at (0.25, 4.5) {};
		\node [style=none] (243) at (0.25, 4.5) {$\overline{Z_1}$};
		\node [style=none] (244) at (4, 3.5) {};
		\node [style=none] (247) at (4, 5.75) {};
		\node [style=none] (248) at (6.25, 6) {};
		\node [style=none] (249) at (7.75, 6) {};
		\node [style=none] (250) at (7.75, 3.25) {};
		\node [style=none] (251) at (6.25, 3.25) {};
		\node [style=none] (252) at (7.75, 5.75) {};
		\node [style=none] (255) at (7.75, 3.5) {};
		\node [style=none] (256) at (6.25, 4) {};
		\node [style=none] (257) at (6.25, 5.25) {};
		\node [style=none] (258) at (7, 4.5) {$E$};
		\node [style=none] (259) at (5.5, 5.25) {};
		\node [style=none] (260) at (5.5, 4) {};
		\node [style=none] (261) at (8.5, 6) {};
		\node [style=none] (262) at (9.75, 6) {};
		\node [style=none] (263) at (9.75, 3.25) {};
		\node [style=none] (264) at (8.5, 3.25) {};
		\node [style=none] (265) at (9.75, 5.75) {};
		\node [style=none] (268) at (9.75, 3.5) {};
		\node [style=none] (271) at (9.25, 4.5) {$\overline{\mathcal{Z}_1}$};
		\node [style=none] (272) at (8.5, 5.75) {};
		\node [style=none] (274) at (8.5, 3.5) {};
		\node [style=none] (275) at (10.25, 5.75) {};
		\node [style=none] (278) at (10.25, 3.5) {};
		\node [style=none] (279) at (-2.25, 4.75) {};
		\node [style=none] (280) at (-2.25, 4.75) {and};
		\node [style=none] (281) at (11, 4.25) {.};
		\node [style=none] (282) at (-10.75, 4.75) {$\vdots$};
		\node [style=none] (283) at (-13, 4.75) {$\vdots$};
		\node [style=none] (284) at (-14.75, 4.75) {$\vdots$};
		\node [style=none] (285) at (10, 4.75) {$\vdots$};
		\node [style=none] (286) at (8, 4.75) {$\vdots$};
		\node [style=none] (287) at (6, 4.75) {$\vdots$};
		\node [style=none] (288) at (3.5, 4.75) {$\vdots$};
		\node [style=none] (289) at (1.25, 4.75) {$\vdots$};
		\node [style=none] (290) at (-0.5, 4.75) {$\vdots$};
		\node [style=none] (291) at (-4.25, 4.75) {$\vdots$};
		\node [style=none] (292) at (-6.25, 4.75) {$\vdots$};
		\node [style=none] (293) at (-8.25, 4.75) {$\vdots$};
	\end{pgfonlayer}
	\begin{pgfonlayer}{edgelayer}
		\draw (133.center)
			 to (132.center)
			 to (135.center)
			 to (134.center)
			 to cycle;
		\draw (146.center) to (147.center);
		\draw (147.center) to (148.center);
		\draw (148.center) to (145.center);
		\draw (149.center) to (151.center);
		\draw (150.center) to (152.center);
		\draw (145.center) to (146.center);
		\draw (153.center) to (141.center);
		\draw (154.center) to (140.center);
		\draw (136.center) to (181.center);
		\draw (139.center) to (178.center);
		\draw (186.center) to (185.center);
		\draw (185.center) to (188.center);
		\draw (188.center) to (187.center);
		\draw (187.center) to (186.center);
		\draw (203.center) to (202.center);
		\draw (202.center) to (205.center);
		\draw (205.center) to (204.center);
		\draw (204.center) to (203.center);
		\draw (200.center) to (194.center);
		\draw (201.center) to (193.center);
		\draw (189.center) to (213.center);
		\draw (192.center) to (215.center);
		\draw (206.center) to (216.center);
		\draw (209.center) to (219.center);
		\draw (222.center)
			 to (221.center)
			 to (224.center)
			 to (223.center)
			 to cycle;
		\draw (233.center) to (234.center);
		\draw [in=270, out=90] (234.center) to (235.center);
		\draw (235.center) to (232.center);
		\draw (236.center) to (238.center);
		\draw (237.center) to (239.center);
		\draw (232.center) to (233.center);
		\draw (240.center) to (230.center);
		\draw (241.center) to (229.center);
		\draw (225.center) to (247.center);
		\draw (228.center) to (244.center);
		\draw (249.center) to (248.center);
		\draw (248.center) to (251.center);
		\draw (251.center) to (250.center);
		\draw (250.center) to (249.center);
		\draw (262.center) to (261.center);
		\draw (261.center) to (264.center);
		\draw (264.center) to (263.center);
		\draw (263.center) to (262.center);
		\draw (259.center) to (257.center);
		\draw (260.center) to (256.center);
		\draw (252.center) to (272.center);
		\draw (255.center) to (274.center);
		\draw (265.center) to (275.center);
		\draw (268.center) to (278.center);
	\end{pgfonlayer}
\end{tikzpicture}
$$
\end{definition}

\vspace{-4 em}

In other words, pushing $\overline{X_i} \; (\text{or} \; \overline{Z_i})$ through $E$ yields $\overline{\mathcal{X}_i}$ $(\text{or}\; \overline{\mathcal{Z}_i})$. Using ZX rewrite rules along with the ZX (or XZ) normal form, we can prove the following lemma.
\begin{lemma}
For any CSS code, all $\overline{X_i}$ and $\overline{Z_i}$ are implementable by multiple single-qubit Pauli operators. In other words, all CSS codes have \emph{transversal} $\overline{X_i}$ and $\overline{Z_i}$.
\end{lemma}
\begin{proof}
Consider an arbitrary CSS code. Without loss of generality, represent its encoder $E$ in the ZX normal form following \Cref{def:normalform}. Then proceed by applying the $\pi$-copy' rule on every $\overline{X_i}$ (the X spider with a phase $\pi$ on the left-hand side of the encoder $E$).
\end{proof}
Below we illustrate the proof using the $\squarecode$ code as an example. 

\begin{example}
For the $\squarecode$ code, $\overline{X_1} = X_1 X_2$.
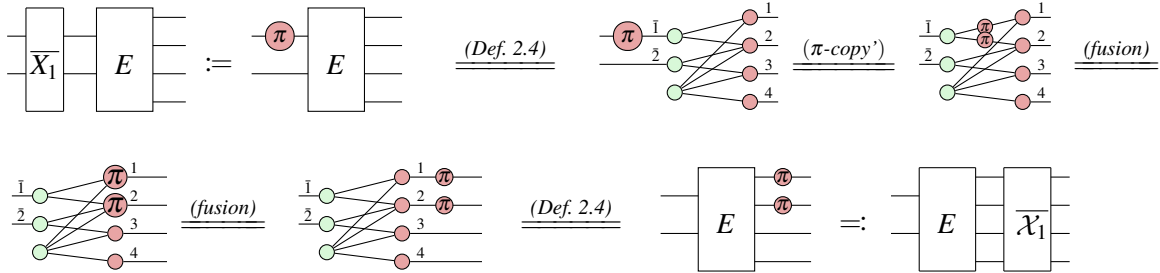
\[
    \begin{tikzpicture}
	\begin{pgfonlayer}{nodelayer}
		\node [style=green small] (0) at (1, 4.5) {};
		\node [style=green small] (5) at (1, 5.25) {};
		\node [style=green small] (8) at (1, 3.75) {};
		\node [style=none] (13) at (-1, 5.25) {};
		\node [style=none] (14) at (-1, 4.5) {};
		\node [style=red small] (15) at (3, 5.75) {};
		\node [style=red small] (16) at (3, 5) {};
		\node [style=red small] (17) at (3, 4.25) {};
		\node [style=red small] (18) at (3, 3.5) {};
		\node [style=none] (19) at (3.75, 5.75) {};
		\node [style=none] (20) at (3.75, 5) {};
		\node [style=none] (21) at (3.75, 4.25) {};
		\node [style=none] (22) at (3.75, 3.5) {};
		\node [style=none] (23) at (3.5, 6) {\tiny $1$};
		\node [style=none] (24) at (3.5, 5.25) {\tiny $2$};
		\node [style=none] (25) at (3.5, 4.5) {\tiny $3$};
		\node [style=none] (26) at (3.5, 3.75) {\tiny $4$};
		\node [style=none] (27) at (0.5, 5.5) {\tiny $\bar{1}$};
		\node [style=none] (28) at (0.5, 4.75) {\tiny $\bar{2}$};
		\node [style=none] (29) at (5.5, 4.75) {$\xlongequal{\scriptsize (\pi\text{-copy'})}$};
		\node [style=none] (33) at (-10.25, 5.25) {};
		\node [style=none] (34) at (-10.25, 4.25) {};
		\node [style=none] (39) at (-6.25, 5.75) {};
		\node [style=none] (40) at (-6.25, 5) {};
		\node [style=none] (41) at (-6.25, 4.25) {};
		\node [style=none] (42) at (-6.25, 3.5) {};
		\node [style=none] (49) at (-3.5, 4.75) {$\xlongequal{\text{(\cref{def:normalform})}}$};
		\node [style=none] (50) at (-8.75, 6) {};
		\node [style=none] (51) at (-7.25, 6) {};
		\node [style=none] (52) at (-7.25, 3.25) {};
		\node [style=none] (53) at (-8.75, 3.25) {};
		\node [style=none] (54) at (-7.25, 5.75) {};
		\node [style=none] (55) at (-7.25, 5) {};
		\node [style=none] (56) at (-7.25, 4.25) {};
		\node [style=none] (57) at (-7.25, 3.5) {};
		\node [style=none] (58) at (-8.75, 4.25) {};
		\node [style=none] (59) at (-8.75, 5.25) {};
		\node [style=none] (60) at (-8, 4.5) {$E$};
		\node [style={red small (text)}] (61) at (-9.5, 5.25) {$\pi$};
		\node [style={red small (text)}] (62) at (-0.25, 5.25) {$\,\pi\,$};
		\node [style=green small] (63) at (8.25, 4.5) {};
		\node [style=green small] (64) at (8.25, 5.25) {};
		\node [style=green small] (65) at (8.25, 3.75) {};
		\node [style=none] (66) at (7.5, 5.25) {};
		\node [style=none] (67) at (7.5, 4.5) {};
		\node [style=red small] (68) at (10.25, 5.75) {};
		\node [style=red small] (69) at (10.25, 5) {};
		\node [style=red small] (70) at (10.25, 4.25) {};
		\node [style=red small] (71) at (10.25, 3.5) {};
		\node [style=none] (72) at (11, 5.75) {};
		\node [style=none] (73) at (11, 5) {};
		\node [style=none] (74) at (11, 4.25) {};
		\node [style=none] (75) at (11, 3.5) {};
		\node [style=none] (76) at (10.75, 6) {\tiny $1$};
		\node [style=none] (77) at (10.75, 5.25) {\tiny $2$};
		\node [style=none] (78) at (10.75, 4.5) {\tiny $3$};
		\node [style=none] (79) at (10.75, 3.75) {\tiny $4$};
		\node [style=none] (80) at (7.75, 5.5) {\tiny $\bar{1}$};
		\node [style=none] (81) at (7.75, 4.75) {\tiny $\bar{2}$};
		\node [style=red small] (82) at (9.25, 5.5) {\tiny$\pi$};
		\node [style=red small] (83) at (9.25, 5.15) {\tiny$\,\pi\,$};
		\node [style=none] (84) at (12.75, 4.75) {$\xlongequal{\scriptsize\text{(fusion)}}$};
		\node [style=green small] (85) at (-8.25, 0.25) {};
		\node [style=green small] (86) at (-8.25, 1) {};
		\node [style=green small] (87) at (-8.25, -0.5) {};
		\node [style=none] (88) at (-9, 1) {};
		\node [style=none] (89) at (-9, 0.25) {};
		\node [style=red small] (90) at (-6.25, 1.5) {};
		\node [style=red small] (91) at (-6.25, 0.75) {};
		\node [style=red small] (92) at (-6.25, 0) {};
		\node [style=red small] (93) at (-6.25, -0.75) {};
		\node [style=none] (94) at (-4.125, 1.5) {};
		\node [style=none] (95) at (-4.125, 0.75) {};
		\node [style=none] (96) at (-4.125, 0) {};
		\node [style=none] (97) at (-4.125, -0.75) {};
		\node [style=none] (98) at (-5.75, 1.75) {\tiny $1$};
		\node [style=none] (99) at (-5.75, 1) {\tiny $2$};
		\node [style=none] (100) at (-5.75, 0.25) {\tiny $3$};
		\node [style=none] (101) at (-5.75, -0.5) {\tiny $4$};
		\node [style=none] (102) at (-8.75, 1.25) {\tiny $\bar{1}$};
		\node [style=none] (103) at (-8.75, 0.5) {\tiny $\bar{2}$};
		\node [style=red small] (104) at (-5.125, 1.5) {\scriptsize$\pi$};
		\node [style=red small] (105) at (-5.125, 0.75) {\scriptsize$\pi$};
		\node [style=none] (106) at (0.625, 1) {};
		\node [style=none] (107) at (0.625, 0) {};
		\node [style=none] (108) at (4.625, 1.5) {};
		\node [style=none] (109) at (4.625, 0.75) {};
		\node [style=none] (110) at (4.625, 0) {};
		\node [style=none] (111) at (4.625, -0.75) {};
		\node [style=none] (113) at (1.625, 1.75) {};
		\node [style=none] (114) at (3.125, 1.75) {};
		\node [style=none] (115) at (3.125, -1) {};
		\node [style=none] (116) at (1.625, -1) {};
		\node [style=none] (117) at (3.125, 1.5) {};
		\node [style=none] (118) at (3.125, 0.75) {};
		\node [style=none] (119) at (3.125, 0) {};
		\node [style=none] (120) at (3.125, -0.75) {};
		\node [style=none] (121) at (1.625, 0) {};
		\node [style=none] (122) at (1.625, 1) {};
		\node [style=none] (123) at (2.375, 0.25) {$E$};
		\node [style=red small] (124) at (3.875, 1.5) {\scriptsize$\pi$};
		\node [style=red small] (125) at (3.875, 0.75) {\scriptsize$\pi$};
		\node [style=none] (132) at (-14.375, 6) {};
		\node [style=none] (133) at (-12.875, 6) {};
		\node [style=none] (134) at (-12.875, 3.25) {};
		\node [style=none] (135) at (-14.375, 3.25) {};
		\node [style=none] (136) at (-12.875, 5.75) {};
		\node [style=none] (137) at (-12.875, 5) {};
		\node [style=none] (138) at (-12.875, 4.25) {};
		\node [style=none] (139) at (-12.875, 3.5) {};
		\node [style=none] (140) at (-14.375, 4.25) {};
		\node [style=none] (141) at (-14.375, 5.25) {};
		\node [style=none] (142) at (-13.625, 4.5) {$E$};
		\node [style=none] (144) at (-11.125, 4.5) {$\coloneqq$};
		\node [style=none] (145) at (-16.25, 6) {};
		\node [style=none] (146) at (-16.25, 3.25) {};
		\node [style=none] (147) at (-15.25, 3.25) {};
		\node [style=none] (148) at (-15.25, 6) {};
		\node [style=none] (149) at (-16.75, 5.25) {};
		\node [style=none] (150) at (-16.75, 4.25) {};
		\node [style=none] (151) at (-16.25, 5.25) {};
		\node [style=none] (152) at (-16.25, 4.25) {};
		\node [style=none] (153) at (-15.25, 5.25) {};
		\node [style=none] (154) at (-15.25, 4.25) {};
		\node [style=none] (155) at (-15.75, 4.5) {};
		\node [style=none] (156) at (-15.75, 4.5) {$\overline{X_1}$};
		\node [style=green small] (157) at (-15.875, 0.25) {};
		\node [style=green small] (158) at (-15.875, 1) {};
		\node [style=green small] (159) at (-15.875, -0.5) {};
		\node [style=none] (160) at (-16.625, 1) {};
		\node [style=none] (161) at (-16.625, 0.25) {};
		\node [style=red small] (162) at (-13.875, 1.5) {$\pi$};
		\node [style=red small] (163) at (-13.875, 0.75) {$\pi$};
		\node [style=red small] (164) at (-13.875, 0) {};
		\node [style=red small] (165) at (-13.875, -0.75) {};
		\node [style=none] (166) at (-12.5, 1.5) {};
		\node [style=none] (167) at (-12.5, 0.75) {};
		\node [style=none] (168) at (-12.5, 0) {};
		\node [style=none] (169) at (-12.5, -0.75) {};
		\node [style=none] (170) at (-13.375, 1.75) {\tiny $1$};
		\node [style=none] (171) at (-13.375, 1) {\tiny $2$};
		\node [style=none] (172) at (-13.375, 0.25) {\tiny $3$};
		\node [style=none] (173) at (-13.375, -0.5) {\tiny $4$};
		\node [style=none] (174) at (-16.375, 1.25) {\tiny $\bar{1}$};
		\node [style=none] (175) at (-16.375, 0.5) {\tiny $\bar{2}$};
		\node [style=none] (178) at (-12, 3.5) {};
		\node [style=none] (179) at (-12, 4.25) {};
		\node [style=none] (180) at (-12, 5) {};
		\node [style=none] (181) at (-12, 5.75) {};
		\node [style=none] (184) at (5.75, 0.25) {$\eqqcolon$};
		\node [style=none] (185) at (7.5, 1.75) {};
		\node [style=none] (186) at (9, 1.75) {};
		\node [style=none] (187) at (9, -1) {};
		\node [style=none] (188) at (7.5, -1) {};
		\node [style=none] (189) at (9, 1.5) {};
		\node [style=none] (190) at (9, 0.75) {};
		\node [style=none] (191) at (9, 0) {};
		\node [style=none] (192) at (9, -0.75) {};
		\node [style=none] (193) at (7.5, 0) {};
		\node [style=none] (194) at (7.5, 1) {};
		\node [style=none] (195) at (8.25, 0.25) {$E$};
		\node [style=none] (200) at (6.75, 1) {};
		\node [style=none] (201) at (6.75, 0) {};
		\node [style=none] (202) at (9.75, 1.75) {};
		\node [style=none] (203) at (11, 1.75) {};
		\node [style=none] (204) at (11, -1) {};
		\node [style=none] (205) at (9.75, -1) {};
		\node [style=none] (206) at (11, 1.5) {};
		\node [style=none] (207) at (11, 0.75) {};
		\node [style=none] (208) at (11, 0) {};
		\node [style=none] (209) at (11, -0.75) {};
		\node [style=none] (210) at (9.75, 0) {};
		\node [style=none] (211) at (9.75, 1) {};
		\node [style=none] (212) at (10.5, 0.25) {$\overline{\mathcal{X}_1}$};
		\node [style=none] (213) at (9.75, 1.5) {};
		\node [style=none] (214) at (9.75, 0.75) {};
		\node [style=none] (215) at (9.75, -0.75) {};
		\node [style=none] (216) at (11.5, 1.5) {};
		\node [style=none] (217) at (11.5, 0.75) {};
		\node [style=none] (218) at (11.5, 0) {};
		\node [style=none] (219) at (11.5, -0.75) {};
		\node [style=none] (220) at (-11, 0.5) {$\xlongequal{\scriptsize\text{(fusion)}}$};
		\node [style=none] (221) at (-1.75, 0.5) {$\xlongequal{\text{(\cref{def:normalform})}}$};
	\end{pgfonlayer}
	\begin{pgfonlayer}{edgelayer}
		\draw (5) to (13.center);
		\draw (0) to (14.center);
		\draw (15) to (5);
		\draw (16) to (5);
		\draw (16) to (0);
		\draw (0) to (17);
		\draw (15) to (8);
		\draw (8) to (16);
		\draw (17) to (8);
		\draw (8) to (18);
		\draw (19.center) to (15);
		\draw (16) to (20.center);
		\draw (21.center) to (17);
		\draw (18) to (22.center);
		\draw (51.center)
			 to (50.center)
			 to (53.center)
			 to (52.center)
			 to cycle;
		\draw (59.center) to (33.center);
		\draw (58.center) to (34.center);
		\draw (39.center) to (54.center);
		\draw (40.center) to (55.center);
		\draw (41.center) to (56.center);
		\draw (42.center) to (57.center);
		\draw (64) to (66.center);
		\draw (63) to (67.center);
		\draw (68) to (64);
		\draw (69) to (64);
		\draw (69) to (63);
		\draw (63) to (70);
		\draw (68) to (65);
		\draw (65) to (69);
		\draw (70) to (65);
		\draw (65) to (71);
		\draw (72.center) to (68);
		\draw (69) to (73.center);
		\draw (74.center) to (70);
		\draw (71) to (75.center);
		\draw (86) to (88.center);
		\draw (85) to (89.center);
		\draw (90) to (86);
		\draw (91) to (86);
		\draw (91) to (85);
		\draw (85) to (92);
		\draw (90) to (87);
		\draw (87) to (91);
		\draw (92) to (87);
		\draw (87) to (93);
		\draw (94.center) to (90);
		\draw (91) to (95.center);
		\draw (96.center) to (92);
		\draw (93) to (97.center);
		\draw (114.center)
			 to (113.center)
			 to (116.center)
			 to (115.center)
			 to cycle;
		\draw (122.center) to (106.center);
		\draw (121.center) to (107.center);
		\draw (108.center) to (117.center);
		\draw (109.center) to (118.center);
		\draw (110.center) to (119.center);
		\draw (111.center) to (120.center);
		\draw (133.center)
			 to (132.center)
			 to (135.center)
			 to (134.center)
			 to cycle;
		\draw (146.center) to (147.center);
		\draw (147.center) to (148.center);
		\draw (148.center) to (145.center);
		\draw (149.center) to (151.center);
		\draw (150.center) to (152.center);
		\draw (145.center) to (146.center);
		\draw (153.center) to (141.center);
		\draw (154.center) to (140.center);
		\draw (158) to (160.center);
		\draw (157) to (161.center);
		\draw (162) to (158);
		\draw (163) to (158);
		\draw (163) to (157);
		\draw (157) to (164);
		\draw (162) to (159);
		\draw (159) to (163);
		\draw (164) to (159);
		\draw (159) to (165);
		\draw (166.center) to (162);
		\draw (163) to (167.center);
		\draw (168.center) to (164);
		\draw (165) to (169.center);
		\draw (136.center) to (181.center);
		\draw (137.center) to (180.center);
		\draw (138.center) to (179.center);
		\draw (139.center) to (178.center);
		\draw (186.center) to (185.center);
		\draw (185.center) to (188.center);
		\draw (188.center) to (187.center);
		\draw (187.center) to (186.center);
		\draw (203.center) to (202.center);
		\draw (202.center) to (205.center);
		\draw (205.center) to (204.center);
		\draw (204.center) to (203.center);
		\draw (200.center) to (194.center);
		\draw (201.center) to (193.center);
		\draw (189.center) to (213.center);
		\draw (190.center) to (214.center);
		\draw (191.center) to (210.center);
		\draw (192.center) to (215.center);
		\draw (206.center) to (216.center);
		\draw (207.center) to (217.center);
		\draw (208.center) to (218.center);
		\draw (209.center) to (219.center);
	\end{pgfonlayer}
\end{tikzpicture}
.
\]
\end{example}

Beyond just X or Z spiders, one can push \emph{any} ZX diagram acting on the logical qubits through the encoder. Such pushing is bidirectional, and the left-to-right direction is interpreted as finding a physical implementation for a given logical operator.
\begin{proposition}
\label{prop:pushenc}
    Let $E$ be the encoder of a CSS code. For any ZX diagram $L$ on the left-hand side of $E$, one can write down a corresponding ZX diagram $P$ on the right-hand side of $E$, such that $EL=PE$. In other words, $P$ is a valid physical implementation of $L$ for that CSS code.
\label{lem:rewrite}
\end{proposition}
\begin{proof}
We proceed as follows. First, unfuse all spiders on the logical qubit wires of $L$, whenever they are not phase-free or have more than one external wire:
    
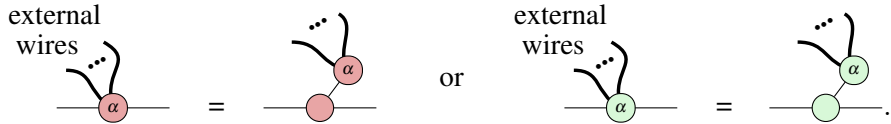
\[
    \begin{tikzpicture}
	\begin{pgfonlayer}{nodelayer}
		\node [style=none] (0) at (-10.5, 0) {};
		\node [style=none] (1) at (-6.25, 0) {=};
		\node [style={red small (text)}] (2) at (-9, 0) {\scriptsize$\alpha$};
		\node [style=none] (3) at (-7.5, 0) {};
		\node [style=none] (4) at (-5, 0) {};
		\node [style=none] (5) at (-2, 0) {};
		\node [style=red] (6) at (-3.5, 0) {};
		\node [style={red small (text)}] (7) at (-2.75, 1) {\scriptsize$\alpha$};
		\node [style=none] (8) at (-3, 2.5) {};
		\node [style=none] (9) at (-4.25, 1.75) {};
		\node [style=none] (10) at (-9.25, 1.75) {};
		\node [style=none] (11) at (-10.25, 1) {};
		\node [style=none] (12) at (-9.425, 1.2) {$\rotatebox{30}{\textbf{...}}$};
		\node [style=none] (13) at (-3.55, 2.225) {$\rotatebox{30}{\textbf{...}}$};
		\node [style=none] (14) at (3, 0) {};
		\node [style=none] (15) at (7.25, 0) {=};
		\node [style={green small (text)}] (16) at (4.5, 0) {\scriptsize$\alpha$};
		\node [style=none] (17) at (6, 0) {};
		\node [style=none] (24) at (4.25, 1.75) {};
		\node [style=none] (25) at (3.25, 1) {};
		\node [style=none] (26) at (4.075, 1.2) {$\rotatebox{30}{\textbf{...}}$};
		\node [style=none] (29) at (2.9, 2.5) {external};
		\node [style=none] (30) at (2.7, 1.7) {wires};
		\node [style=none] (32) at (-10.55, 2.5) {external};
		\node [style=none] (33) at (-10.75, 1.7) {wires};
		\node [style=none] (34) at (8.45, 0) {};
		\node [style=none] (35) at (11.45, 0) {};
		\node [style=green] (36) at (9.95, 0) {};
		\node [style={green small (text)}] (37) at (10.7, 1) {\scriptsize$\alpha$};
		\node [style=none] (38) at (10.45, 2.5) {};
		\node [style=none] (39) at (9.2, 1.75) {};
		\node [style=none] (40) at (9.9, 2.225) {$\rotatebox{30}{\textbf{...}}$};
		\node [style=none] (42) at (0, 0.75) {or};
	\end{pgfonlayer}
	\begin{pgfonlayer}{edgelayer}
		\draw (0.center) to (3.center);
		\draw (4.center) to (5.center);
		\draw [style=edge] (6) to (7);
		\draw [style=edge, line width=1.3pt, in=-30, out=120] (7) to (8.center);
		\draw [style=edge, line width=1.3pt, in=0, out=165, looseness=1.25] (7) to (9.center);
		\draw [style=edge, line width=1.3pt, in=0, out=150, looseness=1.25] (2) to (11.center);
		\draw [style=edge, line width=1.3pt, in=-30, out=105, looseness=1.50] (2) to (10.center);
		\draw (14.center) to (17.center);
		\draw [style=edge, line width=1.3pt, in=0, out=150, looseness=1.25] (16) to (25.center);
		\draw [style=edge, line width=1.3pt, in=-30, out=105, looseness=1.50] (16) to (24.center);
		\draw (34.center) to (35.center);
		\draw [style=edge] (36) to (37);
		\draw [style=edge, line width=1.3pt, in=-30, out=120] (37) to (38.center);
		\draw [style=edge, line width=1.3pt, in=0, out=165, looseness=1.25] (37) to (39.center);
	\end{pgfonlayer}
\end{tikzpicture}
.
    \]
    
    For each X (or Z) spider on the logical qubit wire, rewriting $E$ to be in ZX (or XZ) normal form and applying the strong complementarity (sc) rule yields:
    
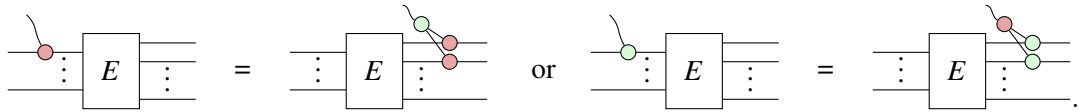
\[
    \begin{tikzpicture}
	\begin{pgfonlayer}{nodelayer}
		\node [style=rec] (0) at (-9.5, 0.75) {$E$};
		\node [style=none] (1) at (-8.75, 1.5) {};
		\node [style=none] (3) at (-8.75, 0) {};
		\node [style=none] (4) at (-7.25, 1.5) {};
		\node [style=none] (6) at (-7.25, 0) {};
		\node [style=none] (9) at (-10.25, 0.25) {};
		\node [style=none] (10) at (-12.25, 0.25) {};
		\node [style=red small] (11) at (-11.25, 1.25) {};
		\node [style=none] (12) at (-11.75, 2.25) {};
		\node [style=rec] (13) at (-2.5, 0.75) {$E$};
		\node [style=none] (14) at (-1.75, 1.5) {};
		\node [style=none] (15) at (-1.75, 1) {};
		\node [style=none] (16) at (-1.75, 0) {};
		\node [style=none] (17) at (0.5, 1.5) {};
		\node [style=none] (18) at (0.5, 1) {};
		\node [style=none] (19) at (0.5, 0) {};
		\node [style=green small] (24) at (-1.25, 2) {};
		\node [style=red small] (25) at (-0.5, 1.5) {};
		\node [style=red small] (26) at (-0.5, 1) {};
		\node [style=none] (27) at (-1.75, 2.5) {};
		\node [style=none] (29) at (-6, 0.75) {=};
		\node [style=none] (30) at (-12.25, 1.25) {};
		\node [style=none] (31) at (-10.25, 1.25) {};
		\node [style=none] (34) at (-3.25, 0.25) {};
		\node [style=none] (35) at (-4.75, 0.25) {};
		\node [style=none] (37) at (-4.75, 1.25) {};
		\node [style=none] (38) at (-3.25, 1.25) {};
		\node [style=none] (39) at (-10.75, 1) {$\vdots$};
		\node [style=none] (40) at (-4, 1) {$\vdots$};
		\node [style=none] (41) at (-8, 0.75) {$\vdots$};
		\node [style=none] (42) at (-1.25, 0.75) {$\vdots$};
		\node [style=none] (54) at (-8.75, 1) {};
		\node [style=none] (55) at (-7.25, 1) {};
		\node [style=rec] (56) at (6, 0.75) {$E$};
		\node [style=none] (57) at (6.75, 1.5) {};
		\node [style=none] (58) at (6.75, 0) {};
		\node [style=none] (59) at (8.25, 1.5) {};
		\node [style=none] (60) at (8.25, 0) {};
		\node [style=none] (61) at (5.25, 0.25) {};
		\node [style=none] (62) at (3.25, 0.25) {};
		\node [style=green small] (63) at (4.25, 1.25) {};
		\node [style=none] (64) at (3.75, 2.25) {};
		\node [style=rec] (65) at (13, 0.75) {$E$};
		\node [style=none] (66) at (13.75, 1.5) {};
		\node [style=none] (67) at (13.75, 1) {};
		\node [style=none] (68) at (13.75, 0) {};
		\node [style=none] (69) at (16, 1.5) {};
		\node [style=none] (70) at (16, 1) {};
		\node [style=none] (71) at (16, 0) {};
		\node [style=red small] (72) at (14.25, 2) {};
		\node [style=green small] (73) at (15, 1.5) {};
		\node [style=green small] (74) at (15, 1) {};
		\node [style=none] (75) at (13.75, 2.5) {};
		\node [style=none] (76) at (9.5, 0.75) {=};
		\node [style=none] (77) at (3.25, 1.25) {};
		\node [style=none] (78) at (5.25, 1.25) {};
		\node [style=none] (79) at (12.25, 0.25) {};
		\node [style=none] (80) at (10.75, 0.25) {};
		\node [style=none] (81) at (10.75, 1.25) {};
		\node [style=none] (82) at (12.25, 1.25) {};
		\node [style=none] (83) at (4.75, 1) {$\vdots$};
		\node [style=none] (84) at (11.5, 1) {$\vdots$};
		\node [style=none] (85) at (7.5, 0.75) {$\vdots$};
		\node [style=none] (86) at (14.25, 0.75) {$\vdots$};
		\node [style=none] (87) at (6.75, 1) {};
		\node [style=none] (88) at (8.25, 1) {};
		\node [style=none] (89) at (2, 0.75) {or};
	\end{pgfonlayer}
	\begin{pgfonlayer}{edgelayer}
		\draw [style=edge] (10.center) to (9.center);
		\draw [style=edge] (1.center) to (4.center);
		\draw [style=edge] (3.center) to (6.center);
		\draw [style=edge, in=-45, out=120, looseness=1.25] (11) to (12.center);
		\draw [style=edge] (14.center) to (17.center);
		\draw [style=edge] (15.center) to (18.center);
		\draw [style=edge] (16.center) to (19.center);
		\draw [style=edge] (25) to (24);
		\draw [style=edge] (26) to (24);
		\draw [style=edge, in=0, out=150] (24) to (27.center);
		\draw (30.center) to (31.center);
		\draw [style=edge] (35.center) to (34.center);
		\draw (37.center) to (38.center);
		\draw [style=edge] (54.center) to (55.center);
		\draw [style=edge] (62.center) to (61.center);
		\draw [style=edge] (57.center) to (59.center);
		\draw [style=edge] (58.center) to (60.center);
		\draw [style=edge, in=-45, out=120, looseness=1.25] (63) to (64.center);
		\draw [style=edge] (66.center) to (69.center);
		\draw [style=edge] (67.center) to (70.center);
		\draw [style=edge] (68.center) to (71.center);
		\draw [style=edge] (73) to (72);
		\draw [style=edge] (74) to (72);
		\draw [style=edge, in=0, out=150] (72) to (75.center);
		\draw (77.center) to (78.center);
		\draw [style=edge] (80.center) to (79.center);
		\draw (81.center) to (82.center);
		\draw [style=edge] (87.center) to (88.center);
	\end{pgfonlayer}
\end{tikzpicture}
.
    \]
    On the left-hand side, a phase-free X (or Z) spider acts on the $i$-th logical qubit; on the right-hand side, phase-free X (or Z) spiders act on all physical qubits wherever $\overline{X}_i$ (or $\overline{Z}_i$) has support. Therefore, any type of $L$ can be pushed through $E$, resulting in a diagram $P$ which satisfies $EL = PE$.
\end{proof}

In \cite{GarvieL2017verifying832}, it was proved that a physical implementation $P$ of a logical operator $L$ satisfies $L = E^{\dagger} P E$. This is implied by $EL = PE$ as $E^{\dagger} E = I$.

  \section{Graphical Morphing of CSS Codes}
\label{sec:morphing}

One way to transform CSS codes is known as \emph{code morphing}. It provides a systematic framework to construct new codes from an existing code while preserving the number of logical qubits in the morphed code. Here, we present this procedure through the rewrites of the encoder diagram using the ZX calculus. Let us start by revisiting the code morphing definition in~\cite{VasmerM2022morphingcodes}.

\begin{definition}
\label{def:codemorphing}
Let $\mathcal{S}$ be a stabilizer group and $\mathcal{C}$ be its joint $+1$ eigenspace. $\mathcal{C}$ is called the \emph{parent code}. Let $Q$ denote the set of physical qubits of $\mathcal{C}$ and $R \subseteq Q$. Then $\mathcal{S}(R)$ is a subgroup of $\mathcal{S}$ generated by all stabilizers of $\mathcal{S}$ that are fully supported on $R$. Let $\mathcal{C}(R)$ be the joint $+1$ eigenspace of $\mathcal{S}(R)$, and $\mathcal{C}(R)$ is called the \emph{child code}. Given the parent code encoder $E_{\mathcal{C}}$, concatenate it with the inverse of the child code encoder $E^\dagger_{\mathcal{C}(R)}$. This gives the \emph{morphed code} $\mathcal{C}_{\setminus R}$.
\end{definition}

\Cref{fig:codetransformvsmorphing} provides two equivalent interpretations for the code morphing process. In \cref{fig:morphcircuit}, \Cref{def:codemorphing} is depicted by the circuit diagram. Since $E_{\mathcal{C}(R)}$ is an isometry, $E^\dagger_{\mathcal{C}(R)}E_{\mathcal{C}(R)}=I$. By construction, the equation shown in \cref{fig:morphcircuit} holds~\cite{VasmerM2022morphingcodes}. Moreover, the parameters of $\mathcal{C}=\llbracket n,k,d\rrbracket$, $\mathcal{C}(R)=\llbracket n_1,k_1,d_1 \rrbracket$, and $\mathcal{C}_{\setminus R}=\llbracket n_2,k_2,d_2 \rrbracket$ are characterized below. Let $m, m_1, m_2$ be the number of stabilizer generators for $\mathcal{C}$, $\mathcal{C}(R)$, and $\mathcal{C}_{\setminus R}$ respectively. Then
\[
n_2=n-n_1+k_1, \quad k_2 = k, \quad m_2 = (n-k)-(n_1-k_1) = m - m_1, \quad d_1, d_2 \in \N.
\]

\Cref{fig:morphzx} provides a concrete example of applying \Cref{def:codemorphing} to the $\steanecode$ Steane code, where $S = \{1,2,3,4,5,6,7\}$ and $R = \{2,3,6,7\}$. As a result, the $\fivecode$ code is morphed from the parent code along with the $\squarecode$ child code. This morphed code inherits a fault-tolerant implementation of the Clifford group from the $\steanecode$ code, which has a transversal implementation of the logical Clifford operators. This morphing process is represented in the ZX diagram by cutting the edges labelled by $\overline{1}$ and $\overline{2}$ adjacent to the X spider. This is equivalent to concatenating the ZX diagram of $E^\dagger_{\squarecode}$ in \cref{fig:morphcircuit}.

\begin{figure}[H]
     \centering
     \begin{subfigure}[b]{1\textwidth}
         \centering
         \scalebox{.8}{\begin{tikzpicture}
	\begin{pgfonlayer}{nodelayer}
		\node [style=none] (31) at (-8.5, 1.875) {};
		\node [style=none] (32) at (-8.5, 0.375) {};
		\node [style=none] (33) at (-0.75, 2) {};
		\node [style=none] (34) at (-0.75, 0.25) {};
		\node [style=none] (35) at (-3.75, -0.5) {};
		\node [style=none] (36) at (-3.75, -2) {};
		\node [style=none] (37) at (-7.5, 2.25) {};
		\node [style=none] (38) at (-4.25, 2.25) {};
		\node [style=none] (39) at (-4.25, -2.5) {};
		\node [style=none] (40) at (-7.5, -2.5) {};
		\node [style=none] (41) at (-4.25, 2) {};
		\node [style=none] (42) at (-4.25, 0.25) {};
		\node [style=none] (43) at (-4.5, -0.5) {};
		\node [style=none] (44) at (-4.5, -1.75) {};
		\node [style=none] (45) at (-7.5, 0.375) {};
		\node [style=none] (46) at (-7.5, 1.875) {};
		\node [style=none] (47) at (-5.675, 1) {$E_\mathcal{C}$};
		\node [style=none] (48) at (-6.5, -0.25) {};
		\node [style=none] (49) at (-6.5, -2.25) {};
		\node [style=none] (50) at (-4.5, -2.25) {};
		\node [style=none] (51) at (-4.5, -0.25) {};
		\node [style=none] (52) at (-7, -0.5) {};
		\node [style=none] (53) at (-6.5, -0.5) {};
		\node [style=none] (54) at (-7, -1.75) {};
		\node [style=none] (55) at (-6.5, -1.75) {};
		\node [style=none] (56) at (-5.35, -1.375) {$E_{\mathcal{C}(R)}$};
		\node [style=none] (57) at (-8, 1.375) {$\vdots$};
		\node [style=none] (58) at (-2.75, 1.25) {$\vdots$};
		\node [style=none] (59) at (-4, -0.875) {$\vdots$};
		\node [style=none] (60) at (-6.75, -0.925) {$\vdots$};
		\node [style=none] (64) at (-3.75, -0.5) {};
		\node [style=none] (65) at (-3.75, -1.75) {};
		\node [style=none] (66) at (-1.5, -0.25) {};
		\node [style=none] (67) at (-1.5, -2.25) {};
		\node [style=none] (68) at (-3.75, -2.25) {};
		\node [style=none] (69) at (-3.75, -0.25) {};
		\node [style=none] (70) at (-1, -0.5) {};
		\node [style=none] (71) at (-1.5, -0.5) {};
		\node [style=none] (72) at (-1, -1.75) {};
		\node [style=none] (73) at (-1.5, -1.75) {};
		\node [style=none] (74) at (-2.65, -1.125) {$E_{\mathcal{C}(R)}^{\dagger}$};
		\node [style=none] (76) at (-1.25, -0.925) {$\vdots$};
		\node [style=none] (77) at (0.25, 0) {$=$};
		\node [style=none] (78) at (1, 1.875) {};
		\node [style=none] (79) at (1, 0.375) {};
		\node [style=none] (80) at (6, 2) {};
		\node [style=none] (81) at (6, 0.25) {};
		\node [style=none] (84) at (2, 2.25) {};
		\node [style=none] (85) at (5, 2.25) {};
		\node [style=none] (86) at (5, -2.5) {};
		\node [style=none] (87) at (2, -2.5) {};
		\node [style=none] (88) at (5, 2) {};
		\node [style=none] (89) at (5, 0.25) {};
		\node [style=none] (92) at (2, 0.375) {};
		\node [style=none] (93) at (2, 1.875) {};
		\node [style=none] (94) at (3.575, 0.75) {$E_{\mathcal{C}_{\setminus R}}$};
		\node [style=none] (99) at (5, -0.5) {};
		\node [style=none] (101) at (5, -1.75) {};
		\node [style=none] (104) at (1.5, 1.375) {$\vdots$};
		\node [style=none] (105) at (5.5, 1.25) {$\vdots$};
		\node [style=none] (106) at (5.5, -0.95) {$\vdots$};
		\node [style=none] (114) at (6, -0.5) {};
		\node [style=none] (116) at (6, -1.75) {};
	\end{pgfonlayer}
	\begin{pgfonlayer}{edgelayer}
		\draw (38.center) to (37.center);
		\draw (37.center) to (40.center);
		\draw (40.center) to (39.center);
		\draw (39.center) to (38.center);
		\draw (46.center) to (31.center);
		\draw (45.center) to (32.center);
		\draw (33.center) to (41.center);
		\draw (34.center) to (42.center);
		\draw (50.center)
			 to (49.center)
			 to (48.center)
			 to (51.center)
			 to cycle;
		\draw (53.center) to (52.center);
		\draw (54.center) to (55.center);
		\draw (68.center)
			 to (67.center)
			 to (66.center)
			 to (69.center)
			 to cycle;
		\draw (71.center) to (70.center);
		\draw (72.center) to (73.center);
		\draw (65.center) to (44.center);
		\draw (43.center) to (64.center);
		\draw (85.center) to (84.center);
		\draw (84.center) to (87.center);
		\draw (87.center) to (86.center);
		\draw (86.center) to (85.center);
		\draw (93.center) to (78.center);
		\draw (92.center) to (79.center);
		\draw (80.center) to (88.center);
		\draw (81.center) to (89.center);
		\draw (114.center) to (99.center);
		\draw (101.center) to (116.center);
	\end{pgfonlayer}
\end{tikzpicture}}
         \caption{Code morphing in the circuit diagram}
         \label{fig:morphcircuit}
     \end{subfigure}
     \hfill
     \begin{subfigure}[b]{1\textwidth}
         \centering
         \scalebox{.75}{\begin{tikzpicture}
	\begin{pgfonlayer}{nodelayer}
		\node [style=red small, label={left:\small 1}] (0) at (7.5, 4.5) {};
		\node [style=red small, label={above:\small 2}] (1) at (15, 0) {};
		\node [style=red small, label={above:\small 3}] (2) at (12.75, 2.5) {};
		\node [style=red small, label={left:\small 4}] (3) at (8, -3.5) {};
		\node [style=red small, label={left:\small 5}] (4) at (7.5, 0.5) {};
		\node [style=red small, label={above:\small 6}] (5) at (13, -2.75) {};
		\node [style=red small, label={left:\small 7}] (6) at (11.75, 0) {};
		\node [style=green small] (7) at (13.5, 0) {};
		\node [style=green small] (8) at (11, 2.5) {};
		\node [style=green small] (9) at (11, -2.5) {};
		\node [style=green small] (10) at (6.25, 0.5) {};
		\node [style=none] (11) at (4.5, 0.5) {};
		\node [style=none] (12) at (16.5, 0) {};
		\node [style=none] (13) at (16.5, -2.75) {};
		\node [style=none] (14) at (16.5, -3.5) {};
		\node [style=none] (15) at (16.5, 3.75) {};
		\node [style=none] (16) at (16.5, 4.5) {};
		\node [style=none] (18) at (16.5, 2.5) {};
		\node [style=green small] (19) at (8.5, 2) {};
		\node [style=green small] (20) at (8.5, -1) {};
		\node [style=none] (21) at (16, 3) {};
		\node [style=none] (22) at (16, -3) {};
		\node [style=none] (23) at (10.375, -3) {};
		\node [style=none] (24) at (10.375, 3) {};
		\node [style=none] (25) at (10.125, 5.25) {};
		\node [style=none] (26) at (10.125, -4.25) {};
		\node [style=none] (27) at (5.75, -4.25) {};
		\node [style=none] (28) at (5.75, 5.25) {};
		\node [style=none] (50) at (3.5, 1) {$\xlongequal[\scriptsize(\text{fusion})]{\small(\text{id})}$};
		\node [style=red small, label={above:\small $\overline{1}$}] (105) at (9.75, 2.25) {};
		\node [style=red small, label={above:\small $\overline{2}$}] (106) at (9.75, -1.75) {};
		\node [style=red small, label={above:\small 1}] (107) at (-3.75, 4.25) {};
		\node [style=red small, label={above:\small 2}] (108) at (1.25, 0.5) {};
		\node [style=red small, label={above:\small 3}] (109) at (-1, 2) {};
		\node [style=red small, label={left:\small 4}] (110) at (-3.75, -3.75) {};
		\node [style=red small, label={left:\small 5}] (111) at (-3.75, 0.5) {};
		\node [style=red small, label={above:\small 6}] (112) at (-1, -1) {};
		\node [style=red small, label={left:\small 7}] (113) at (-2, 0.5) {};
		\node [style=green small] (114) at (-0.25, 0.5) {};
		\node [style=green small] (115) at (-5.75, 0.5) {};
		\node [style=none] (116) at (-7.25, 0.5) {};
		\node [style=none] (117) at (2.5, 0.5) {};
		\node [style=none] (118) at (2.5, -1) {};
		\node [style=none] (119) at (2.5, -3.75) {};
		\node [style=none] (121) at (2.5, 4.25) {};
		\node [style=none] (122) at (2.5, -2.5) {};
		\node [style=none] (123) at (2.5, 2) {};
		\node [style=green small] (124) at (-2.75, 2) {};
		\node [style=green small] (125) at (-2.75, -1) {};
		\node [style=none] (126) at (-3, -4.75) {$E_{\mathcal{C}}$};
		\node [style=none] (127) at (2.5, 3.25) {};
		\node [style=none] (128) at (16.5, -1.25) {};
		\node [style=red small, label={above:\small 1}] (130) at (21, 4.5) {};
		\node [style=red small, label={above:\small 2}] (131) at (33, 1) {};
		\node [style=red small, label={above:\small 3}] (132) at (29.5, 3.5) {};
		\node [style=red small, label={left:\small 4}] (133) at (21, -3.5) {};
		\node [style=red small, label={above:\small 5}] (134) at (21, 0.5) {};
		\node [style=red small, label={above:\small 6}] (135) at (29.5, -1.75) {};
		\node [style=red small, label={above:\small 7}] (136) at (29.5, 1) {};
		\node [style=green small] (137) at (31.5, 1) {};
		\node [style=green small] (138) at (28, 2) {};
		\node [style=green small] (139) at (28, -0.25) {};
		\node [style=green small] (140) at (19.5, 0.5) {};
		\node [style=none] (141) at (18, 0.5) {};
		\node [style=none] (142) at (34, 1) {};
		\node [style=none] (143) at (34, -1.75) {};
		\node [style=none] (144) at (25.75, -3.5) {};
		\node [style=none] (145) at (25.75, 0.5) {};
		\node [style=none] (146) at (25.75, 4.5) {};
		\node [style=none] (147) at (34, 3.5) {};
		\node [style=green small] (149) at (22.5, -1.5) {};
		\node [style=none] (159) at (22.875, -4.75) {$E_{\mathcal{C}_{\setminus R}}$};
		\node [style=none] (162) at (34, -0.5) {};
		\node [style=none] (163) at (17, 1) {$\xrightarrow{\mathbf{CUT}}$};
		\node [style=green small] (164) at (22.5, 2.25) {};
		\node [style=red small, label={above:\small$\overline{1}$}] (167) at (24, 2.25) {};
		\node [style=red small, label={above:\small$\overline{2}$}] (168) at (24, -1.5) {};
		\node [style=none] (169) at (25.75, 2.25) {};
		\node [style=none] (170) at (25.75, -1.5) {};
		\node [style=none] (171) at (26.25, 2) {};
		\node [style=none] (172) at (26.25, -0.25) {};
		\node [style=none] (173) at (8, -4.75) {$E_{\mathcal{C}_{\setminus R}}$};
		\node [style=none] (174) at (-6.75, 5) {};
		\node [style=none] (175) at (2, 5) {};
		\node [style=none] (176) at (2, -4.25) {};
		\node [style=none] (177) at (-6.75, -4.25) {};
		\node [style=none] (178) at (31.5, -4.75) {$E_{\mathcal{C}(R)}$};
		\node [style=none] (179) at (13.5, -4.75) {$E_{\mathcal{C}(R)}$};
	\end{pgfonlayer}
	\begin{pgfonlayer}{edgelayer}
		\draw [style=edge] (1) to (7);
		\draw [style=edge] (7) to (5);
		\draw [style=edge] (7) to (6);
		\draw [style=edge] (7) to (2);
		\draw [style=edge] (2) to (8);
		\draw [style=edge] (8) to (6);
		\draw [style=edge] (9) to (6);
		\draw [style=edge] (9) to (5);
		\draw [style=edge] (0) to (10);
		\draw [style=edge] (10) to (4);
		\draw [style=edge] (10) to (3);
		\draw [style=edge] (11.center) to (10);
		\draw [style=edge] (1) to (12.center);
		\draw [style=edge] (5) to (13.center);
		\draw [style=edge] (3) to (14.center);
		\draw [style=edge, in=165, out=120, looseness=0.75] (4) to (15.center);
		\draw [style=edge] (0) to (16.center);
		\draw [style=edge] (2) to (18.center);
		\draw (9) to (20);
		\draw (20) to (3);
		\draw (20) to (4);
		\draw (4) to (19);
		\draw (19) to (8);
		\draw (19) to (0);
		\draw [style=blue box] (21.center)
			 to (22.center)
			 to (23.center)
			 to (24.center)
			 to [in=180, out=0] cycle;
		\draw [style=violet box] (28.center)
			 to (25.center)
			 to (26.center)
			 to (27.center)
			 to cycle;
		\draw [style=edge] (108) to (114);
		\draw [style=edge] (114) to (112);
		\draw [style=edge] (114) to (113);
		\draw [style=edge] (114) to (109);
		\draw [style=edge] (107) to (115);
		\draw [style=edge] (115) to (111);
		\draw [style=edge] (115) to (110);
		\draw [style=edge] (116.center) to (115);
		\draw [style=edge] (108) to (117.center);
		\draw [style=edge] (112) to (118.center);
		\draw [style=edge] (110) to (119.center);
		\draw [style=edge, in=180, out=-90, looseness=1.50] (113) to (122.center);
		\draw [style=edge] (107) to (121.center);
		\draw [style=edge] (109) to (123.center);
		\draw (125) to (110);
		\draw (125) to (111);
		\draw (111) to (124);
		\draw (124) to (107);
		\draw [style=edge] (109) to (124);
		\draw [style=edge] (124) to (113);
		\draw [style=edge] (125) to (113);
		\draw [style=edge] (125) to (112);
		\draw [style=edge] (131) to (137);
		\draw [style=edge] (137) to (135);
		\draw [style=edge] (137) to (136);
		\draw [style=edge] (137) to (132);
		\draw [style=edge] (132) to (138);
		\draw [style=edge] (138) to (136);
		\draw [style=edge] (139) to (136);
		\draw [style=edge] (139) to (135);
		\draw [style=edge] (130) to (140);
		\draw [style=edge] (140) to (134);
		\draw [style=edge] (140) to (133);
		\draw [style=edge] (141.center) to (140);
		\draw [style=edge] (131) to (142.center);
		\draw [style=edge] (135) to (143.center);
		\draw [style=edge] (133) to (144.center);
		\draw [style=edge] (130) to (146.center);
		\draw [style=edge] (132) to (147.center);
		\draw (149) to (133);
		\draw (149) to (134);
		\draw (134) to (145.center);
		\draw (130) to (164);
		\draw (164) to (134);
		\draw (164) to (167);
		\draw (149) to (168);
		\draw (167) to (169.center);
		\draw (168) to (170.center);
		\draw (138) to (171.center);
		\draw (139) to (172.center);
		\draw [style=edge, in=165, out=90] (111) to (127.center);
		\draw [style=edge, in=195, out=-90, looseness=0.75] (6) to (128.center);
		\draw [style=edge, in=-165, out=-90, looseness=0.75] (136) to (162.center);
		\draw [style=red box] (174.center) to (175.center);
		\draw [style=red box] (175.center) to (176.center);
		\draw [style=red box] (176.center) to (177.center);
		\draw [style=red box] (174.center) to (177.center);
	\end{pgfonlayer}
\end{tikzpicture}}
         \caption{Code morphing of the Steane code in the ZX diagram}
         \label{fig:morphzx}
     \end{subfigure}
     \caption{Code morphing can be visualized using both circuit and ZX diagrams. In \cref{fig:morphcircuit}, code morphing is viewed as a concatenation of the parent code encoder $E_{\mathcal{C}}$ and the inverse of the child code encoder $E^\dagger_{\mathcal{C}(R)}$. In \Cref{fig:morphzx}, the encoder $E_\mathcal{C}$ of the Steane code is represented in the ZX normal form. As described in \Cref{proc:codemorphzx}, by applying ZX rules (id) and (fusion) in \cref{fig:zxrules}, we can perform code morphing by bipartitioning it into the encoder $E_{\mathcal{C}_{\setminus R}}$ of the morphed code $\mathcal{C}_{\setminus R} = 
     \fivecode$, and the encoder $E_{\mathcal{C}(R)}$ of the child code $\mathcal{C}(R) = \squarecode$.}
        \label{fig:codetransformvsmorphing}
\end{figure}
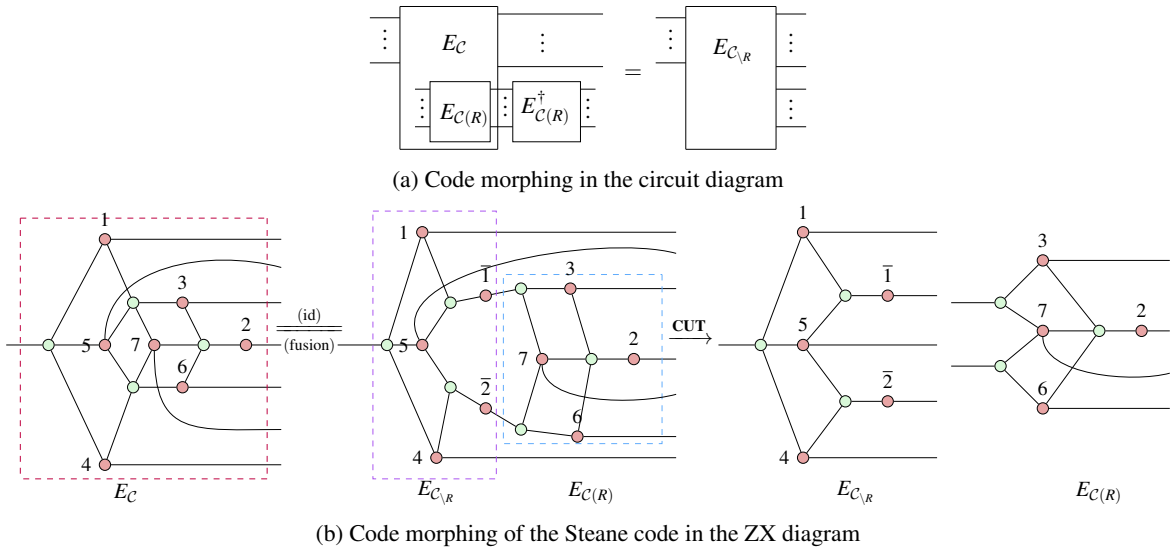

Next, we generalize the notion of code morphing and show how ZX calculus could be used to study these relations between the encoders of different CSS codes. More precisely, we provide an algorithm to morph a new CSS code from an existing CSS code.

\begin{procedure}
\label{proc:codemorphzx}
    Given a parent code $\mathcal{C}$ and a child code $\mathcal{C}(R)$ satisfying \Cref{def:codemorphing}, construct the encoder of $\mathcal{C}$ in the \ZXNF. Then the code morphing proceeds as follows:
    \begin{enumerate}[label={(\alph*)}]
        \item Unfuse every Z spider which is supported on $c$ qubits within $R$ and $f$ qubits outside $R$, $c \neq 0$, $f \neq 0$. 
        \item Add an identity X spider between each pair of Z spiders being unfused in step \textit{(a)}.
        \item Cut the edge between every identity X spider and the Z spiders supported on the $f$ qubits in $R$.
    \end{enumerate}
\end{procedure}
It follows that the subdiagram containing $R$ corresponds to the ZX normal form of $E_{\mathcal{C}(R)}$. It has the same number of X spiders as R, so $n_1=\lvert R\rvert$. Suppose that there are $h$ Z spiders being unfused. Then $h$ must be bounded by the number of Z spiders in the ZX normal form of $E_{\mathcal{C}}$. As each spider unfusion introduces a logical qubit to $\mathcal{C}(R)$, $k_1 = h$.  On the other hand, the complement subdiagram contains $n-n_1 + k_1$ X spiders as each edge cut introduces a new X spider into the complement subdiagram. It also contains $k$ logical qubits as the input edges in the ZX normal form of $E_{\mathcal{C}}$ are invariant throughout the spider-unfusing and edge-cutting process. This gives the ZX normal form for the encoder of the morphed code $\mathcal{C_{\setminus R}}=\llbracket n_2,k_2,d_2 \rrbracket$, where $n_2 = n-n_1 + k_1$, $k_2 = k$, $d_2\in \N$. As a result, the ZX normal form of $E_{\mathcal{C}}$ is decomposed into the ZX normal forms of $E_{\mathcal{C}(R)}$ and $E_{\mathcal{C_{\setminus R}}}$ respectively.

As the XZ and ZX normal forms are equivalent for CSS codes, \Cref{proc:codemorphzx} can be carried out for the XZ normal form by inverting the roles of Z and X at each step.

Here, we exemplify the application of \Cref{proc:codemorphzx} by morphing two simple CSS codes. Unlike \Cref{fig:morphzx}, \Cref{exm:1} chooses a different subset of qubits, $R=\{4,5,6,7\}$, to obtain the $\llbracket 6,1,1 \rrbracket$ morphed code. In \Cref{exm:2}, we visualize the $\tencode$ code morphing from the $\fifteencode$ quantum Reed-Muller code. The $\tencode$ code is interesting because it inherits a fault-tolerant implementation of the logical $T$ gate from its parent code, which has a transversal implementation of the logical $T$ gate.

\begin{example}
    Let the parent code $\mathcal{C}$ be the Steane code and the child code be $C(R) = \llbracket 4,3,1 \rrbracket$. By \Cref{proc:codemorphzx}, we obtain the morphed code $\mathcal{C}_{\setminus R}=\llbracket 6,1,1 \rrbracket$. Note that for $C(R)$, there is one X-type stabilizer generator and no Z-type stabilizer generator. This means that $C(R)$ cannot detect a single-qubit $X$ error, so it has a distance of $1$. In $\mathcal{C}_{\setminus R}$, the physical qubit labelled $\overline{3}$ is not protected by any X-type stabilizer. Therefore, $\mathcal{C}_{\setminus R}$ is of distance $1$.
        
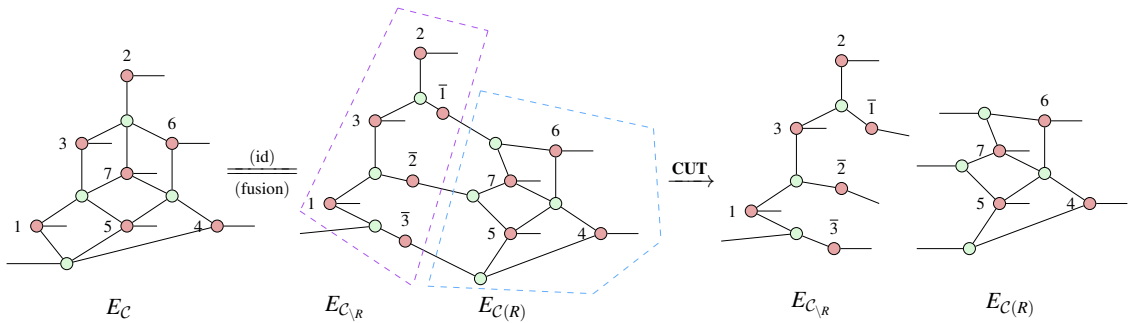
\begin{figure}[H]
        \centering
        \scalebox{.8}{\begin{tikzpicture}
	\begin{pgfonlayer}{nodelayer}
		\node [style=red small, label={left:\scriptsize 1}] (29) at (-19, -1) {};
		\node [style=red small, label={above:\scriptsize 2}] (30) at (-16, 4) {};
		\node [style=red small, label={left:\scriptsize 3}] (31) at (-17.5, 1.75) {};
		\node [style=red small, label={left:\scriptsize 4}] (32) at (-13, -1) {};
		\node [style=red small, label={left:\scriptsize 5}] (33) at (-16, -1) {};
		\node [style=red small, label={above:\scriptsize 6}] (34) at (-14.5, 1.75) {};
		\node [style=red small, label={left:\scriptsize 7}] (35) at (-16, 0.75) {};
		\node [style=green small] (36) at (-16, 2.5) {};
		\node [style=green small] (39) at (-18, -2.25) {};
		\node [style=none] (40) at (-20, -2.25) {};
		\node [style=none] (41) at (-14.75, 4) {};
		\node [style=none] (42) at (-13.25, 1.75) {};
		\node [style=none] (43) at (-11.75, -1) {};
		\node [style=none] (44) at (-15, -1) {};
		\node [style=none] (45) at (-18, -1) {};
		\node [style=none] (46) at (-15, 0.75) {};
		\node [style=none] (47) at (-16.5, 1.75) {};
		\node [style=green small] (48) at (-17.5, 0) {};
		\node [style=green small] (49) at (-14.5, 0) {};
		\node [style=none] (50) at (-11.5, 0.75) {$\xlongequal[\scriptsize(\text{fusion})]{\small(\text{id})}$};
		\node [style=none] (102) at (13.375, -3.625) {$E_{\mathcal{C}(R)}$};
		\node [style=none] (103) at (6.75, -3.625) {$E_{\mathcal{C}_{\setminus R}}$};
		\node [style=none] (104) at (-16.25, -3.75) {$E_{\mathcal{C}}$};
		\node [style=none] (145) at (2.75, 0.75) {$\xrightarrow{\mathbf{CUT}}$};
		\node [style=red small, label={left:\scriptsize 1}] (190) at (-9.25, -0.25) {};
		\node [style=red small, label={above:\scriptsize 2}] (191) at (-6.25, 4.75) {};
		\node [style=red small, label={left:\scriptsize 3}] (192) at (-7.75, 2.5) {};
		\node [style=red small, label={left:\scriptsize 4}] (193) at (-0.25, -1.25) {};
		\node [style=red small, label={left:\scriptsize 5}] (194) at (-3.25, -1.25) {};
		\node [style=red small, label={above:\scriptsize 6}] (195) at (-1.75, 1.5) {};
		\node [style=red small, label={left:\scriptsize 7}] (196) at (-3.25, 0.5) {};
		\node [style=green small] (197) at (-6.25, 3.25) {};
		\node [style=green small] (198) at (-4.25, -2.75) {};
		\node [style=none] (199) at (-10.25, -1.25) {};
		\node [style=none] (200) at (-5, 4.75) {};
		\node [style=none] (201) at (-0.5, 1.5) {};
		\node [style=none] (202) at (1, -1.25) {};
		\node [style=none] (203) at (-2.25, -1.25) {};
		\node [style=none] (204) at (-8.25, -0.25) {};
		\node [style=none] (205) at (-2.25, 0.5) {};
		\node [style=none] (206) at (-6.75, 2.5) {};
		\node [style=green small] (207) at (-7.75, 0.75) {};
		\node [style=green small] (208) at (-1.75, -0.25) {};
		\node [style=green small] (209) at (-3.75, 1.75) {};
		\node [style=green small] (210) at (-4.5, 0) {};
		\node [style=green small] (211) at (-7.75, -1) {};
		\node [style=red small, label={above:\scriptsize $\overline{1}$}] (212) at (-5.5, 2.75) {};
		\node [style=red small, label={above:\scriptsize $\overline{2}$}] (213) at (-6.5, 0.5) {};
		\node [style=red small, label={above:\scriptsize $\overline{3}$}] (214) at (-6.75, -1.5) {};
		\node [style=none] (215) at (-7, 6.25) {};
		\node [style=none] (216) at (-10.25, -0.5) {};
		\node [style=none] (217) at (-6.5, -3) {};
		\node [style=none] (218) at (-4, 5.5) {};
		\node [style=none] (219) at (-4.25, 3.5) {};
		\node [style=none] (220) at (-6, -3) {};
		\node [style=none] (221) at (-0.5, -3.25) {};
		\node [style=none] (222) at (1.5, 2.25) {};
		\node [style=none] (223) at (1.75, -1.5) {};
		\node [style=red small, label={left:\scriptsize 1}] (225) at (4.75, -0.5) {};
		\node [style=red small, label={above:\scriptsize 2}] (226) at (7.75, 4.5) {};
		\node [style=red small, label={left:\scriptsize 3}] (227) at (6.25, 2.25) {};
		\node [style=red small, label={left:\scriptsize 4}] (228) at (16, -0.25) {};
		\node [style=red small, label={left:\scriptsize 5}] (229) at (13, -0.25) {};
		\node [style=red small, label={above:\scriptsize 6}] (230) at (14.5, 2.5) {};
		\node [style=red small, label={left:\scriptsize 7}] (231) at (13, 1.5) {};
		\node [style=green small] (232) at (7.75, 3) {};
		\node [style=green small] (233) at (12, -1.75) {};
		\node [style=none] (234) at (3.75, -1.5) {};
		\node [style=none] (235) at (9, 4.5) {};
		\node [style=none] (236) at (15.75, 2.5) {};
		\node [style=none] (237) at (17.25, -0.25) {};
		\node [style=none] (238) at (14, -0.25) {};
		\node [style=none] (239) at (5.75, -0.5) {};
		\node [style=none] (240) at (14, 1.5) {};
		\node [style=none] (241) at (7.25, 2.25) {};
		\node [style=green small] (242) at (6.25, 0.5) {};
		\node [style=green small] (243) at (14.5, 0.75) {};
		\node [style=green small] (244) at (12.5, 2.75) {};
		\node [style=green small] (245) at (11.75, 1) {};
		\node [style=green small] (246) at (6.25, -1.25) {};
		\node [style=red small, label={above:\scriptsize $\overline{1}$}] (247) at (8.75, 2.25) {};
		\node [style=red small, label={above:\scriptsize $\overline{2}$}] (248) at (7.75, 0.25) {};
		\node [style=red small, label={above:\scriptsize $\overline{3}$}] (249) at (7.5, -1.75) {};
		\node [style=none] (250) at (10, 2) {};
		\node [style=none] (251) at (9, -0.25) {};
		\node [style=none] (252) at (8.75, -1.75) {};
		\node [style=none] (253) at (11, 2.75) {};
		\node [style=none] (254) at (10.25, 1) {};
		\node [style=none] (255) at (10.25, -1.75) {};
		\node [style=none] (256) at (-8.75, -3.75) {$E_{\mathcal{C}_{\setminus R}}$};
		\node [style=none] (257) at (-3.5, -3.75) {$E_{\mathcal{C}(R)}$};
	\end{pgfonlayer}
	\begin{pgfonlayer}{edgelayer}
		\draw [style=edge] (30) to (36);
		\draw [style=edge] (36) to (34);
		\draw [style=edge] (36) to (35);
		\draw [style=edge] (36) to (31);
		\draw [style=edge] (29) to (39);
		\draw [style=edge] (39) to (33);
		\draw [style=edge] (39) to (32);
		\draw [style=edge] (40.center) to (39);
		\draw [style=edge] (30) to (41.center);
		\draw [style=edge] (34) to (42.center);
		\draw [style=edge] (32) to (43.center);
		\draw [style=edge] (33) to (44.center);
		\draw [style=edge] (35) to (46.center);
		\draw [style=edge] (29) to (45.center);
		\draw [style=edge] (31) to (47.center);
		\draw (49) to (32);
		\draw (49) to (33);
		\draw (33) to (48);
		\draw (48) to (29);
		\draw [style=edge] (31) to (48);
		\draw [style=edge] (48) to (35);
		\draw [style=edge] (49) to (35);
		\draw [style=edge] (49) to (34);
		\draw [style=edge] (191) to (197);
		\draw [style=edge] (197) to (192);
		\draw [style=edge] (190) to (198);
		\draw [style=edge] (198) to (194);
		\draw [style=edge] (198) to (193);
		\draw [style=edge] (191) to (200.center);
		\draw [style=edge] (195) to (201.center);
		\draw [style=edge] (193) to (202.center);
		\draw [style=edge] (194) to (203.center);
		\draw [style=edge] (196) to (205.center);
		\draw [style=edge] (190) to (204.center);
		\draw [style=edge] (192) to (206.center);
		\draw (208) to (193);
		\draw (208) to (194);
		\draw (207) to (190);
		\draw [style=edge] (192) to (207);
		\draw [style=edge] (208) to (196);
		\draw [style=edge] (208) to (195);
		\draw (195) to (209);
		\draw (196) to (209);
		\draw (197) to (209);
		\draw (207) to (210);
		\draw (196) to (210);
		\draw (210) to (194);
		\draw (199.center) to (211);
		\draw [style=violet box] (215.center) to (216.center);
		\draw [style=violet box] (216.center) to (217.center);
		\draw [style=violet box] (215.center) to (218.center);
		\draw [style=violet box] (218.center) to (217.center);
		\draw [style=blue box] (219.center) to (220.center);
		\draw [style=blue box] (219.center) to (222.center);
		\draw [style=blue box] (220.center) to (221.center);
		\draw [style=blue box] (222.center) to (223.center);
		\draw [style=blue box] (221.center) to (223.center);
		\draw [style=edge] (226) to (232);
		\draw [style=edge] (232) to (227);
		\draw [style=edge] (233) to (229);
		\draw [style=edge] (233) to (228);
		\draw [style=edge] (226) to (235.center);
		\draw [style=edge] (230) to (236.center);
		\draw [style=edge] (228) to (237.center);
		\draw [style=edge] (229) to (238.center);
		\draw [style=edge] (231) to (240.center);
		\draw [style=edge] (225) to (239.center);
		\draw [style=edge] (227) to (241.center);
		\draw (243) to (228);
		\draw (243) to (229);
		\draw (242) to (225);
		\draw [style=edge] (227) to (242);
		\draw [style=edge] (243) to (231);
		\draw [style=edge] (243) to (230);
		\draw (230) to (244);
		\draw (231) to (244);
		\draw (231) to (245);
		\draw (245) to (229);
		\draw (234.center) to (246);
		\draw (232) to (247);
		\draw (242) to (248);
		\draw (225) to (246);
		\draw (246) to (249);
		\draw (247) to (250.center);
		\draw (248) to (251.center);
		\draw (249) to (252.center);
		\draw (244) to (253.center);
		\draw (245) to (254.center);
		\draw (233) to (255.center);
	\end{pgfonlayer}
\end{tikzpicture}}
    \end{figure}
    \label{exm:1}
\end{example}

\begin{example}
    Let the parent code $\mathcal{C}$ be the quantum Reed-Muller and the child code be $\mathcal{C}(R)=\cubecode$. By \Cref{proc:codemorphzx}, we obtain the morphed code $\mathcal{C}_{\setminus R}=\tencode$. For brevity, the X spiders representing physical qubits and the logical qubit wires inputting to the Z spiders are omitted.
    
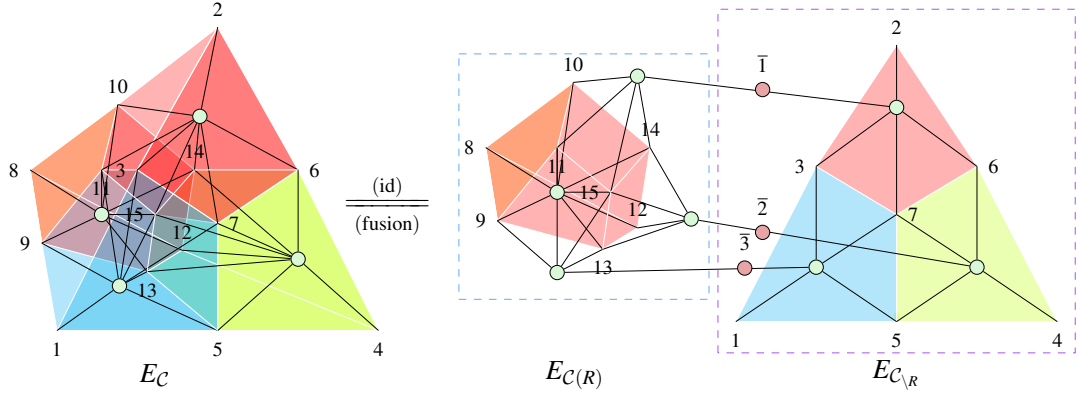
\begin{figure}[H]
        \centering
        \scalebox{.95}{\input{figs/morph/morphing15qrm.tikz}}
    \end{figure}
    \label{exm:2}
\end{example}

 \section{Graphical Code Switching of CSS Codes}
\label{sec:codeswitching}

Another way to transform CSS codes is known as \emph{code switching}. It is a widely studied technique in quantum error correction. Codes with complementary fault-tolerant gate sets are switched between each other to realize a universal set of logical operations. As a case study, we focus on the code switching protocol between the Steane code and the quantum Reed-Muller code~\cite{anderson2014fault,paetznick2013univftqctransversal,QuanDX2018gaugefixRMconvert}. Since this process is bidirectional, the reasoning for one direction can be simply adjusted for the opposite direction. Recall in \Cref{lem:encodedstate}, we showed that the extended Steane code is equivalent to the Steane code up to some auxiliary state. In what follows, we focus on the \emph{backward switching} from the quantum Reed-Muller code to 
the extended Steane code.

Using the ZX calculus, we provide a graphical interpretation for the backward code switching. More precisely, it is visualized as gauge-fixing the $\fifteensub$ subsystem code, followed by a sequence of syndrome-determined recovery operations.

We first characterize the relations between the quantum Reed-Muller code, the extended Steane code, and the $\fifteensub$ subsystem code. For brevity, we denote these codes as $\cqrm, \ \cex$ and $\csub$, and their respective encoders as $\eqrm, \ \eex$, and $\esub$.

\begin{lemma}
\label{lem:reduce}
When the three gauge qubits are in the $\ket{\overline{+++}}$ state, $\csub$ is equal to $\cex$, as shown in \Cref{fig:gaugefixingcircuit}.
\end{lemma}

\begin{figure}[H]
\centering
\scalebox{.9}{\begin{tikzpicture}
	\begin{pgfonlayer}{nodelayer}
		\node [style=none] (0) at (-1.75, 4.75) {};
		\node [style=none] (2) at (1.5, 2.75) {};
		\node [style=none] (3) at (1.5, 1.25) {};
		\node [style=none] (4) at (1.5, 5) {};
		\node [style=none] (5) at (1.5, 0.5) {};
		\node [style=none] (6) at (-1.75, 2) {};
		\node [style=none] (7) at (-1.75, 5) {};
		\node [style=none] (8) at (1.5, 0.5) {};
		\node [style=none] (9) at (-0.05, 3.6) {$\esub$};
		\node [style=none] (10) at (1.5, 4.5) {};
		\node [style=none] (11) at (-4, 4.75) {$|+\rangle$};
		\node [style=none] (12) at (-3.5, 4.75) {};
		\node [style=none] (13) at (-1.75, 4) {};
		\node [style=none] (15) at (-4, 4) {$|+\rangle$};
		\node [style=none] (16) at (-3.5, 4) {};
		\node [style=none] (17) at (-1.75, 3.25) {};
		\node [style=none] (18) at (-4, 3.25) {$|+\rangle$};
		\node [style=none] (19) at (-3.5, 3.25) {};
		\node [style=none] (20) at (-1.75, 2.5) {};
		\node [style=none] (22) at (-4.5, 2.5) {};
		\node [style=none] (23) at (3.25, 4.5) {};
		\node [style=none] (24) at (3.25, 1.25) {};
		\node [style=none] (25) at (4.25, 3) {$=$};
		\node [style=none] (26) at (-5, 4) {$\Biggl \{$};
		\node [style=none] (27) at (-7.5, 4) {\small Gauge qubits};
		\node [style=none] (28) at (-7.5, 2.5) {\small Logical qubit};
		\node [style=none] (29) at (2.75, 3) {$\vdots$};
		\node [style=none] (31) at (11.25, 2.75) {};
		\node [style=none] (32) at (11.25, 1.25) {};
		\node [style=none] (33) at (11.25, 5) {};
		\node [style=none] (34) at (11.25, 0.5) {};
		\node [style=none] (35) at (8, 3) {};
		\node [style=none] (36) at (8, 5) {};
		\node [style=none] (37) at (11.25, 0.5) {};
		\node [style=none] (38) at (9.7, 3.6) {$\eex$};
		\node [style=none] (39) at (11.25, 4.5) {};
		\node [style=none] (48) at (8, 4) {};
		\node [style=none] (49) at (5.25, 4) {};
		\node [style=none] (50) at (13, 4.5) {};
		\node [style=none] (51) at (13, 1.25) {};
		\node [style=none] (53) at (12.5, 3) {$\vdots$};
	\end{pgfonlayer}
	\begin{pgfonlayer}{edgelayer}
		\draw (4.center) to (8.center);
		\draw (8.center) to (6.center);
		\draw (6.center) to (7.center);
		\draw (7.center) to (4.center);
		\draw [in=180, out=0] (12.center) to (0.center);
		\draw [in=180, out=0] (16.center) to (13.center);
		\draw [in=180, out=0] (19.center) to (17.center);
		\draw [in=180, out=0] (22.center) to (20.center);
		\draw (10.center) to (23.center);
		\draw (3.center) to (24.center);
		\draw (33.center) to (37.center);
		\draw (37.center) to (35.center);
		\draw (35.center) to (36.center);
		\draw (36.center) to (33.center);
		\draw [in=180, out=0] (49.center) to (48.center);
		\draw (39.center) to (50.center);
		\draw (32.center) to (51.center);
	\end{pgfonlayer}
\end{tikzpicture}}
\caption{$\csub$ is equivalent to $\cex$ up to a fixed state of gauge qubits.}
\label{fig:gaugefixingcircuit}
\end{figure}
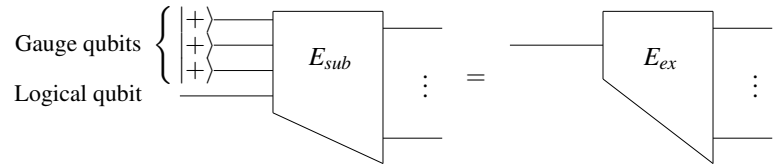

\begin{proof}
According to \cref{def:subsystem}, represent $\esub$ in the XZ normal form, with Z-type stabilizer generators $\SZ{i}$, Z-type gauge operators $\LGZ{j}$, and one logical Z operator $\overline{Z}$, $\ 1 \leq i \leq 7, \ 1 \leq j \leq 3$. After applying a sequence of rewrite rules, we obtain exactly the XZ normal form for $\eex$.
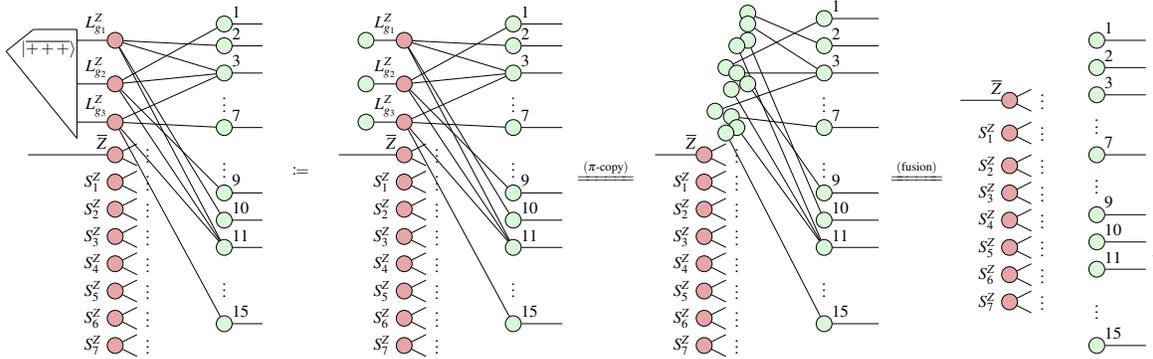
$$\scalebox{0.58}{\input{figs/conv/eqrm4-0.tikz}}.$$

Alternatively, if one chooses to represent $\esub$ in the ZX normal form, the proof proceeds by applying the (fusion) rule to the Z spiders and identifying the gauge operators $\LGX{1}, \ \LGX{2}, \ \LGX{3}$ of $\csub$ as the stabilizers $\SX{5}, \ \SX{6}, \ \SX{7}$ of $\cex$, respectively:
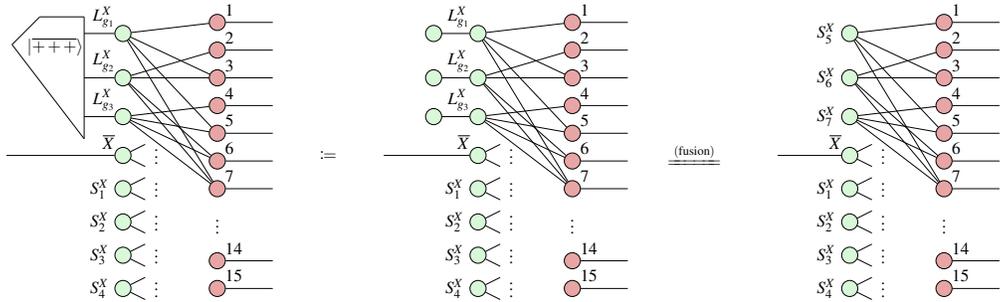
$$\scalebox{0.59}{\input{figs/conv/eqrm4-0zx.tikz}}.$$
\end{proof}

\begin{corollary}
    When the three gauge qubits are in the $\ket{\overline{000}}$ state, $\csub$ is equal to $\cqrm$.
\end{corollary}

In \cite{anderson2014fault,paetznick2013univftqctransversal}, code switching is described as a \emph{gauge fixing} process. Further afield, \cite{vuillot2019code} provides a generic recipe to gauge-fix a CSS subsystem code. Here, we generalize \Cref{lem:reduce} and describe how to gauge-fix $\csub$ to $\cex$ using the ZX calculus. 

\begin{proposition}
\label{prop:gaugefixsub}
Gauge-fixing $\csub$ in the following steps results in $\cex$, as shown in \Cref{fig:generalcir}.

\begin{enumerate}[label={(\alph*)}]
    \item Measure three X-type gauge operators $L_{g_i}^X$ and obtain the corresponding outcomes $k_1,k_2,k_3\in \Z_2$.
    \item When $k_i = 1$, the gauge qubit $i$ has collapsed to the wrong state $\ket{\overline{-}}$. Apply the Z-type recovery operation $L_{g_i}^Z$.
\end{enumerate}
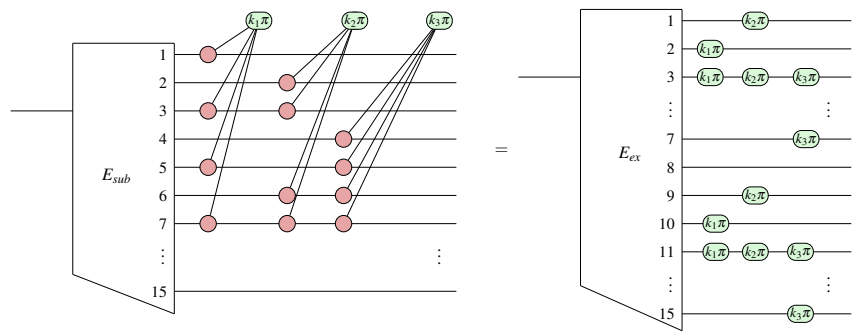
\begin{figure}[H]
\centering
    \scalebox{.6}{\begin{tikzpicture}
	\begin{pgfonlayer}{nodelayer}
		\node [style=none, label={west:\small2}] (0) at (-6.75, 5.75) {};
		\node [style=none, label={west:\small15}] (1) at (-6.75, -3.5) {};
		\node [style=none] (2) at (-6.75, 7.5) {};
		\node [style=none] (3) at (-6.75, -4.5) {};
		\node [style=none] (4) at (-11.25, -2.75) {};
		\node [style=none] (5) at (-11.25, 7.5) {};
		\node [style=none] (6) at (-6.75, -4.5) {};
		\node [style=none] (7) at (-9.3, 1.6) {$\esub$};
		\node [style=none, label={west:\small 1}] (8) at (-6.75, 7) {};
		\node [style=none] (9) at (-11.25, 4.5) {};
		\node [style=none] (10) at (-14, 4.5) {};
		\node [style=none] (11) at (5.75, 7) {};
		\node [style=none] (12) at (5.75, -3.5) {};
		\node [style=none] (13) at (5, -1.75) {$\vdots$};
		\node [style=none] (18) at (5.75, 5.75) {};
		\node [style=none, label={west:\small3}] (19) at (-6.75, 4.5) {};
		\node [style=none] (20) at (5.75, 4.5) {};
		\node [style=none, label={west:\small4}] (21) at (-6.75, 3.25) {};
		\node [style=none] (22) at (5.75, 3.25) {};
		\node [style=none, label={west:\small5}] (23) at (-6.75, 2) {};
		\node [style=none] (24) at (5.75, 2) {};
		\node [style=none, label={west:\small6}] (25) at (-6.75, 0.75) {};
		\node [style=none] (26) at (5.75, 0.75) {};
		\node [style=none, label={west:\small7}] (27) at (-6.75, -0.5) {};
		\node [style=none] (28) at (5.75, -0.5) {};
		\node [style=red] (29) at (-5.25, 7) {};
		\node [style=red] (30) at (-5.25, -0.5) {};
		\node [style=red] (31) at (-5.25, 4.5) {};
		\node [style=red] (32) at (-5.25, 2) {};
		\node [style={green small (text)}] (34) at (-3, 8.5) {$k_1\pi$};
		\node [style=red] (35) at (-1.75, 5.75) {};
		\node [style=red] (36) at (-1.75, 4.5) {};
		\node [style=red] (37) at (-1.75, 0.75) {};
		\node [style=red] (38) at (-1.75, -0.5) {};
		\node [style={green small (text)}] (39) at (1.25, 8.5) {};
		\node [style={green small (text)}] (40) at (1.25, 8.5) {$k_2\pi$};
		\node [style=red] (41) at (0.75, 2) {};
		\node [style=red] (42) at (0.75, 0.75) {};
		\node [style=red] (43) at (0.75, -0.5) {};
		\node [style=red] (44) at (0.75, 3.25) {};
		\node [style={green small (text)}] (46) at (5, 8.5) {$k_3\pi$};
		\node [style=none] (48) at (7.75, 2.75) {$=$};
		\node [style=none] (103) at (-7.25, -1.75) {$\vdots$};
		\node [style=none, label={west:\small2}] (104) at (15.75, 7.25) {};
		\node [style=none, label={west:\small15}] (105) at (15.75, -4.5) {};
		\node [style=none] (106) at (15.75, 9) {};
		\node [style=none] (107) at (15.75, 1) {};
		\node [style=none] (108) at (11.25, -3.25) {};
		\node [style=none] (109) at (11.25, 9) {};
		\node [style=none] (110) at (15.75, -5.25) {};
		\node [style=none] (111) at (13.45, 2.6) {$\eex$};
		\node [style=none, label={west:\small 1}] (112) at (15.75, 8.5) {};
		\node [style=none] (113) at (11.25, 6) {};
		\node [style=none] (114) at (23.25, 8.5) {};
		\node [style=none] (115) at (23.25, -4.5) {};
		\node [style=none] (116) at (23.25, 7.25) {};
		\node [style=none, label={west:\small3}] (117) at (15.75, 6) {};
		\node [style=none] (118) at (23.25, 6) {};
		\node [style=none, label={west:\small7}] (125) at (15.75, 3.25) {};
		\node [style=none] (126) at (23.25, 3.25) {};
		\node [style={green small (text)}] (127) at (17, 7.25) {$\small k_1\pi$};
		\node [style={green small (text)}] (128) at (17, 6) {$\small k_1\pi$};
		\node [style={green small (text)}] (129) at (17.25, -1.75) {$\small k_1\pi$};
		\node [style={green small (text)}] (130) at (17.25, -0.5) {$\small k_1\pi$};
		\node [style={green small (text)}] (131) at (19, 8.5) {$\small k_2\pi$};
		\node [style={green small (text)}] (132) at (19, 6) {$\small k_2\pi$};
		\node [style={green small (text)}] (133) at (19, -1.75) {$\small k_2\pi$};
		\node [style={green small (text)}] (134) at (19, 0.75) {$\small k_2\pi$};
		\node [style={green small (text)}] (135) at (21, -4.5) {$\small k_3\pi$};
		\node [style={green small (text)}] (136) at (21, -1.75) {$\small k_3\pi$};
		\node [style={green small (text)}] (137) at (21.25, 6) {$\small k_3\pi$};
		\node [style={green small (text)}] (138) at (21.25, 3.25) {$\small k_3\pi$};
		\node [style=none] (139) at (11.25, 6) {};
		\node [style=none] (140) at (8.5, 6) {};
		\node [style=none, label={west:\small8}] (143) at (15.75, 2) {};
		\node [style=none] (144) at (23.25, 2) {};
		\node [style=none, label={west:\small9}] (145) at (15.75, 0.75) {};
		\node [style=none] (146) at (23.25, 0.75) {};
		\node [style=none, label={west:\small10}] (147) at (15.75, -0.5) {};
		\node [style=none] (148) at (23.25, -0.5) {};
		\node [style=none, label={west:\small11}] (149) at (15.75, -1.75) {};
		\node [style=none] (150) at (23.25, -1.75) {};
		\node [style=none] (151) at (22.25, -3) {$\vdots$};
		\node [style=none] (152) at (15.25, -3) {$\vdots$};
		\node [style=none] (153) at (22.25, 4.75) {$\vdots$};
		\node [style=none] (154) at (15.25, 4.75) {$\vdots$};
	\end{pgfonlayer}
	\begin{pgfonlayer}{edgelayer}
		\draw (2.center) to (6.center);
		\draw (6.center) to (4.center);
		\draw (4.center) to (5.center);
		\draw (5.center) to (2.center);
		\draw [in=180, out=0] (10.center) to (9.center);
		\draw (8.center) to (11.center);
		\draw (1.center) to (12.center);
		\draw (0.center) to (18.center);
		\draw (19.center) to (20.center);
		\draw (21.center) to (22.center);
		\draw (23.center) to (24.center);
		\draw (25.center) to (26.center);
		\draw (27.center) to (28.center);
		\draw (29) to (34);
		\draw (31) to (34);
		\draw (32) to (34);
		\draw (30) to (34);
		\draw (35) to (40);
		\draw (36) to (40);
		\draw (37) to (40);
		\draw (38) to (40);
		\draw (41) to (46);
		\draw (42) to (46);
		\draw (43) to (46);
		\draw (44) to (46);
		\draw (106.center) to (110.center);
		\draw (110.center) to (108.center);
		\draw (108.center) to (109.center);
		\draw (109.center) to (106.center);
		\draw (112.center) to (114.center);
		\draw (105.center) to (115.center);
		\draw (104.center) to (116.center);
		\draw (117.center) to (118.center);
		\draw (125.center) to (126.center);
		\draw [in=180, out=0] (140.center) to (139.center);
		\draw (143.center) to (144.center);
		\draw (145.center) to (146.center);
		\draw (147.center) to (148.center);
		\draw (149.center) to (150.center);
	\end{pgfonlayer}
\end{tikzpicture}}
    \caption{Gauge-fixing $\csub$ to $\cex$ in the circuit diagram.}
    \label{fig:generalcir}
\end{figure}
\end{proposition}

\begin{proof}
By \cref{def:subsystemnormalform}, construct the ZX normal form of $\esub$ in the blue dashed box of (i). Then the three gauge operators $L_{g_i}^X$ are measured in step (a). The subsequent equalities follow from \cref{fig:zxrules,fig:zxadditionalrules}. Next, we observe that the purple dashed box in (iii) is exactly the encoder of the $\intcode$ stabilizer code. By Lemma 3.2 in \cite{KissingerA2022phasefreeCSS}, it can be equivalently expressed in the XZ normal form, as in (iv). By \Cref{prop:pushenc}, pushing each Z spider with the phase $k_i\pi$ across $E_{\intcode}$ results in (v). In step (b), Pauli Z operators are applied based upon the measurement outcome $k_i$, which corresponds to the recovery operations in the red dashed box of (v). After that, the gauge qubits of $\csub$ are set to the $\ket{\overline{+++}}$ state. By \Cref{lem:reduce}, we obtain the XZ normal form for $\eex$, as shown in the orange dashed box of (vi). Therefore, the equation in \Cref{fig:generalcir} holds.
\end{proof}
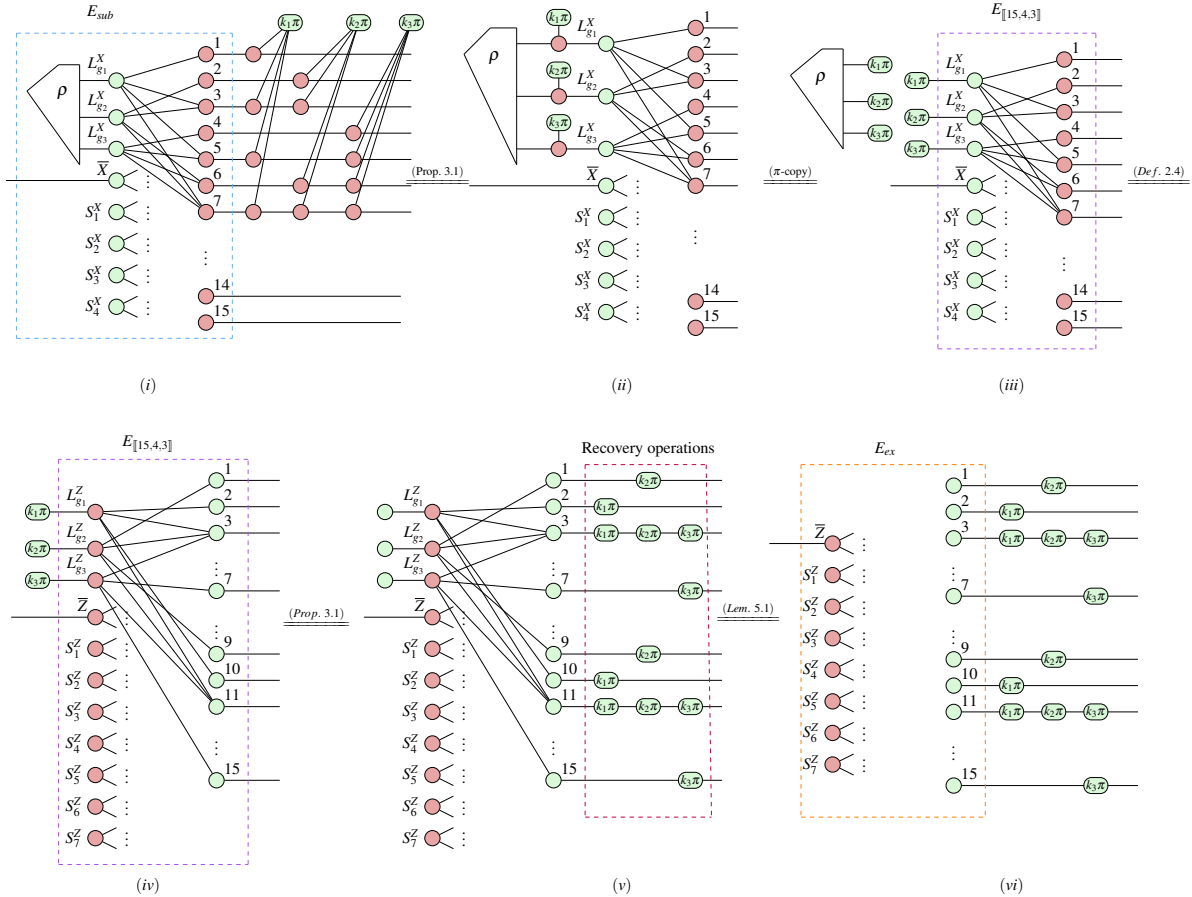
\[
\scalebox{.56}{\input{figs/conv/proofdetails1.tikz}}
\]

We sum up by explaining how to obtain $\cex$ and $\cqrm$ by gauge-fixing $\csub$. In \cref{prop:gaugefixsub}, we showed that measuring the X-type gauge operators $L_{g_i}^X$ followed by the Z-type recovery operations $L_{g_i}^Z$ is equivalent to adding $L_{g_i}^X$ to the stabilizer group $\mathcal{S}_{sub}$. This results in the formation of $\cex$. Analogously, measuring the Z-type gauge operators $L_{g_i}^Z$ followed by the X-type recovery operations $L_{g_i}^X$ is equivalent to adding $L_{g_i}^Z$ to $\mathcal{S}_{sub}$. Thus, we obtain $\cqrm$. 

Alternatively, gauge-fixing $\csub$ can be viewed as a way of switching between $\cex$ and $\cqrm$~\cite{anderson2014fault,QuanDX2018gaugefixRMconvert}. As an example, in \Cref{fig:alt}, we visualize the measurement of $L_{g_1}^X\coloneqq X_1X_3X_5X_7$ in order to switch from $\cqrm$ to  $\cex$. The effect of measuring other X-type gauge operators can be reasoned analogously. 

By \cref{def:subsystemnormalform}, construct the XZ normal form of $\eqrm$ in (i). Then measure $L_{g_1}^X$ and apply a sequence of rewrite rules to the ZX diagram. In (v), the stabilizer $L_{g_1}^Z\coloneqq Z_2Z_3Z_{10}Z_{11}$ is removed from the stabilizer group $\mathcal{S}_{qrm}$. Meanwhile, the recovery operation can be read off from the graphical derivation: $(Z_{2}Z_{3}Z_{10}Z_{11})^{k_1}=\left(L_{g_1}^Z\right)^{k_1}$, $k_1 \in \Z_2$.

\begin{figure}
    \centering
    \scalebox{.6}{\input{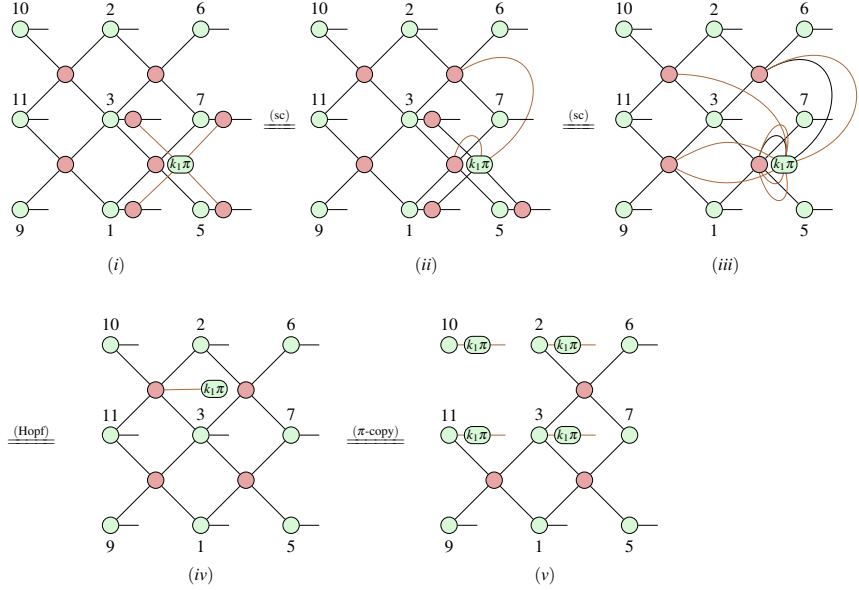}}
    \caption{The switching from $\cqrm$ to $\cex$ provides an alternative interpretation of \Cref{prop:gaugefixsub}. After measuring $L_{g_1}^X$, $L_{g_1}^Z$ is removed from the stabilizer group $\mathcal{S}_{qrm}$ and the recovery operation is performed based on the measurement syndrome. Note that unrelated X and Z spiders are omitted from the ZX diagrams.}
    \label{fig:alt}
\end{figure}

Overall, ZX visualization provides a deeper understanding of the gauge fixing and code switching protocols. On top of revealing the relations between different CSS codes' encoders, it provides a simple yet rigorous test for various fault-tolerant protocols. Beyond this, it will serve as an intuitive guiding principle for the implementation of various logical operations.

\section{Conclusion}
\label{sec:conclude}

In this paper, we generalize the notions in \cite{KissingerA2022phasefreeCSS} and describe a normal form for CSS subsystem codes. Built upon the equivalence between CSS codes and the phase-free ZX diagrams, we provide a bidirectional rewrite rule to establish a correspondence between a logical ZX diagram and its physical implementation. With these tools in place, we provide a graphical representation of two code transformation techniques: code morphing, a procedure that transforms a code through unfusing spiders for the stabilizer generators, and gauge fixing, where different stabilizer codes can be obtained from a common subsystem code. These explicit graphical derivations show how the ZX calculus and graphical encoder maps relate several equivalent perspectives on these code transforming operations, allowing potential utilities of ZX to simplify fault-tolerant protocols and verify their correctness.

Looking ahead, many questions remain. It is still not clear how to present the general code deformation of CSS codes using phase-free ZX diagrams. Besides, understanding code concatenation through the lens of ZX calculus may help derive new and better codes. In addition, it would be interesting to look at other code modification techniques derived from the classical coding theory \cite{macwilliams1977theory}.

\newpage
 \section{Acknowledgement}
\label{sec:acknowledgement}
The authors would like to thank Thomas Scruby for enlightening discussions. SML and MM  wish to thank NTT Research for their financial and technical support. This work was supported in part by Canada’s NSERC. Research at IQC is supported in part by the Government of Canada through Innovation, Science and Economic Development Canada. 
Research at Perimeter Institute is supported in part by the Government of Canada through the Department of Innovation, Science and Economic Development Canada and by the Province of Ontario through the Ministry of Colleges and Universities.
LY is supported by an Oxford - Basil Reeve Graduate Scholarship at Oriel College with the Clarendon Fund.

\bibliographystyle{eptcs}
\bibliography{generic}

\appendix

\end{document}